%% file: main.tex
\RequirePackage{amsmath,amssymb}

\documentclass[runningheads]{llncs}

\input{preamble}

\begin{document}

\input{macros}

\title{Implicit Rankings for Verifying Liveness Properties in First-Order Logic}

\author{\empty}
\institute{\empty}
\author{Raz Lotan\,\orcidlink{0009-0008-5883-5082} \and Sharon Shoham\,\orcidlink{0000-0002-7226-3526}}
\institute{Tel Aviv University, Tel Aviv, Israel\\
\email{lotanraz@tauex.tau.ac.il}\\}

\maketitle            

\input{sections/abstract}

\input{sections/intro}

\input{sections/examples}

\input{sections/prelims}

\input{sections/ranks}

\input{sections/constructions}

\input{sections/evaluation}

\input{sections/related}

\bibliography{references}

\newpage

\appendix

\input{appendices/alternative_defs}

\input{appendices/linear_sum_constructors}

\input{appendices/proofs_constructors}

\input{appendices/extras_technical}

\input{appendices/proof_rule}

\input{appendices/deadlock_freedom}

\input{appendices/abstractions}

\end{document}

%% file: preamble.tex
\usepackage{xcolor}
\usepackage{graphicx}
\usepackage{orcidlink}
\usepackage{hyperref}
\usepackage{cleveref}
\usepackage{faktor}
\usepackage{cite}
\usepackage{lmodern}

\spnewtheorem{constructor}[newcounter]{Constructor}{\bfseries}{\itshape}
\Crefformat{constructor}{#2Const. #1#3}
\Crefrangeformat{constructor}{#3Const.~#1#4--#5#2#6}
\Crefmultiformat{constructor}{Const.~#2#1#3}{ and~#2#1#3}{, #2#1#3}{ and~#2#1#3}

\usepackage{tikz}
\usetikzlibrary{arrows}
\usetikzlibrary{positioning}

\usepackage{listings}
\usepackage[T1]{fontenc}
\usepackage[scaled=0.85]{beramono}
\lstdefinestyle{myStyle}{
    belowcaptionskip=1\baselineskip,
    breaklines=true,
    frame=none,
    numbers=none,
    basicstyle=\scriptsize \ttfamily,
keywordstyle=\bfseries\color{red!40!black},
commentstyle=\itshape\color{purple!40!black},  
    mathescape = true,
    identifierstyle=\color{black},
    stepnumber=1,
    firstnumber=1,
    numbers=none,
    numberstyle=\tiny,
    escapechar=@,
}

\lstdefinelanguage{custom}{
	keywords=[1]{individual, function, while, atomic, assume,  if, else, for, in , relation},
	keywordstyle=[1]{\bfseries},
	keywords=[3]{machine,index,array,bool},
	keywordstyle=[3]{\bfseries\color{blue!40!black}},
	comment=[l]{\#}
}

\usepackage{environ}

\newcommand{\repeatclaim}[1]{%
  \begingroup
  \newcommand{\theclaim}{\ref{#1}}%
  \expandafter\expandafter\expandafter\claim
  \csname repclaim@#1\endcsname
  \endclaim
  \endgroup
}

\NewEnviron{repclaim}[1]{%
  \global\expandafter\xdef\csname repclaim@#1\endcsname{%
    \unexpanded\expandafter{\BODY}%
  }%
  \expandafter\claim\BODY\unskip\label{#1}\endclaim
}

%% file: macros.tex
\newif\ifcomments
\commentsfalse

\newif\ifshort
\shortfalse

\newif\ifconference
\conferencetrue

\ifcomments
\newcommand{\sharon}[1]{\textcolor{purple}{SH: #1 }}
\newcommand{\raz}[1]{\textcolor{cyan}{RL: #1 }}
\else
\newcommand{\sharon}[1]{}
\newcommand{\raz}[1]{}
\fi
\newcommand{\commentout}[1]{}

\newcommand{\nat}{\mathbb{N}}
\newcommand{\seq}[1]{{\vec #1}}
\newcommand{\sort}[1]{\mathsf{#1}}

\newcommand{\formula}[0]{\alpha}
\newcommand{\formulaa}[0]{\beta}
\newcommand{\domain}[0]{\mathcal D}
\newcommand{\interp}[0]{\mathcal I}
\newcommand{\struct}{s}
\newcommand{\structset}{\mathrm{struct}}
\newcommand{\assign}{v}
\newcommand{\otherassign}{u}
\newcommand{\assignset}{\mathrm{assign}}
\newcommand{\pairname}[0]{a-structure}%
\newcommand{\pairnames}[0]{a-structures}%
\newcommand{\pairset}[0]{\mathrm{a\text-struct}}

\newcommand{\pre}[0]{0}
\newcommand{\post}[0]{1}

\newcommand{\Tspec}[0]{\mathcal T}

\newcommand{\rankname}[0]{\mathrm{Rk}}
\newcommand{\low}[0]{0}
\newcommand{\high}[0]{1}
\newcommand{\parameters}[0]{parameters}
\newcommand{\reduced}[0]{\varphi_{<}}
\newcommand{\reducedsuper}[1]{\varphi_{<}^{#1}}
\newcommand{\conserved}[0]{\varphi_{\leq}}
\newcommand{\conservedsuper}[1]{\varphi_{\leq}^{#1}}

\newcommand{\pair}[0]{(\struct,\assign)}
\newcommand{\pairlow}[0]{(\struct_\low,\assign_\low)}
\newcommand{\pairhigh}[0]{(\struct_\high,\assign_\high)}
\newcommand{\twopair}[0]{\pairlow,\pairhigh}

\newcommand{\liveprop}{\mathcal P}
\newcommand{\globalinv}[0]{\rho}
\newcommand{\triggerinv}[0]{\phi}
\newcommand{\helpful}[0]{\psi}
\newcommand{\globally}[0]{\square}
\newcommand{\eventually}[0]{\lozenge}
\newcommand{\prop}[0]{\mathcal P}

\newcommand{\orderformula}[0]{\ell}
\newcommand{\immutorder}[0]{\mathrm{order}}

\newcommand{\concatvar}[0]{{\cdot}}
\newcommand{\concatfunc}[0]{{\cdot}}
\newcommand{\bijec}[0]{\sigma}
\newcommand{\insup}[0]{\circ}

\newcommand{\mypara}[1]{\medskip \noindent \emph{#1}}

\newcommand{\FV}{\textit{FV}}
\newcommand{\Vars}{\textit{Vars}}
\newcommand{\terminterp}[3]{{#1}^{#2,#3}}

\newcommand{\skd}{\mathsf{skd}}
\newcommand{\nextf}{\mathsf{next}}
\newcommand{\ptr}{\mathsf{ptr}}
\newcommand{\lt}{\mathsf{lt}}
\newcommand{\priv}{\mathsf{priv}}
\newcommand{\botm}{\mathsf{bot}}
\newcommand{\arr}{a}

%% file: sections/abstract.tex
\begin{abstract}
Liveness properties are traditionally proven using a ranking function that maps system states to some well-founded set. Carrying out such proofs in first-order logic enables automation by SMT solvers.
However, reasoning about many natural ranking functions is beyond reach of existing solvers.
To address this, 
we introduce the notion of implicit rankings --- first-order formulas that soundly approximate the reduction of some  ranking function without defining it explicitly. 
We provide recursive constructors of implicit rankings that can be instantiated and composed to induce a rich family of implicit rankings.
Our constructors use quantifiers to approximate reasoning about useful primitives such as cardinalities of sets and unbounded sums that are not directly expressible in first-order logic.
We demonstrate the effectiveness of our implicit rankings by verifying liveness properties of several intricate examples, including  Dijkstra's $k$-state, $4$-state and $3$-state self-stabilizing protocols.
\end{abstract}

%% file: sections/intro.tex
\section{Introduction}\label{sec:introduction}

Liveness properties of a system 
assert that some desirable behavi\textbf{}or  eventually happens in all executions of the system.
The most common approach to proving liveness properties
is based on the notion of a well-founded ranking.
Such a proof goes by finding a ranking function $f$ from the set of states of the system to some well-founded set $(A, <)$ such that for any transition of the system from state $\struct$ to state $\struct'$ where the desired behavior does not yet occur, the ranking is reduced, i.e.,
$f(\struct') < f(\struct)$. 
Thus, well-foundedness of $A$ 
ensures that there is no infinite execution that 
does not eventually satisfy the property.

Many recent works use first-order logic (FOL) for verifying safety and liveness properties~\cite{lotan2024provingcutoffboundssafety,towards_liveness_proofs,liveness_to_safety,prophecy,paxos_made_epr,modularity_deductive,lvr}.
FOL has been established as a useful tool for modeling and verifying systems, mostly due to the success of automatic solvers that %
answer complicated satisfiability queries in seconds.
One challenge that arises when proving liveness properties in FOL is that well-foundedness is not directly expressible in FOL.
Additionally, many common primitives that are useful for defining ranking functions, such as cardinality of sets, and sums over unbounded domains, cannot be directly captured in FOL.

We present an approach that overcomes these hurdles and 
facilitates carrying out proofs based on ranking functions in FOL. 
Our approach utilizes the observation that it suffices to encode the reduction
in the ranking function without explicitly encoding the ranking function itself. Furthermore, the reduction %
need not be encoded precisely, but can be soundly approximated.
For this purpose, we define the notion of an \emph{implicit ranking}, which can be soundly used in liveness proofs in place of explicit ranking functions.
An implicit ranking consists of a reduction formula, $\reduced$, and a conservation formula, $\conserved$. These are two-state first-order formulas
for which there exists \emph{some} ranking function $f$ mapping the states to elements of \emph{some} well-founded set $(A,<)$ 
such that whenever states $s$ and $s'$ satisfy $\reduced$ we have $f(s')<f(s)$ and whenever $s$, $s'$ satisfy $\conserved$ we have $f(s') \leq f(s)$.
Conservation is needed, for example, in proof rules for verifying liveness under fairness assumptions.

We then introduce recursive constructors for implicit rankings that can be instantiated and composed to induce a rich family of implicit rankings.
\ifshort \else Our constructors are based on familiar notions from order theory, such as pointwise ordering and lexicographic ordering, which can be used to lift and aggregate orders in various ways, adapted to the first-order setting. %
\fi
A key component of our approach is the introduction of domain-based constructors that use quantification over the domain of a state to express aggregation of rankings.
For example we can express 
a pointwise aggregation of given rankings, or a lexicographic aggregation %
based on a given order on the elements in the domain.
When composed, such aggregations can capture complex ranking arguments that are unattainable by existing methods.

Implicit rankings produced by domain-based constructors are sound for finite, albeit unbounded, domains. Such domains are common when reasoning about distributed protocols, concurrent systems, arrays, and so on, where the set of nodes, array indices, etc.\ %
is finite but not fixed. In fact, liveness of such systems often depends on the finiteness of the domain. Thus, %
such constructors allow us to utilize finiteness of the domain in liveness proofs, even though finiteness itself is not definable in FOL.

As notable examples, we use our constructors to produce implicit rankings for the self-stabilization property of Dijkstra's $k$-state, $4$-state and $3$-state protocols~\cite{dijkstra_self_stab}.
In these examples,  domain-based constructors are able to replace reasoning about unbounded sums and set cardinalities in novel ways.
We use the implicit rankings within a sound proof rule to verify the examples, where the premises of the rule are discharged automatically by an SMT solver. 
To the best of our knowledge, this is the first SMT-based verification of the more challenging $3$-state and $4$-state protocols, and it is simpler than existing proofs for such protocols~\cite{certification_3_state_stab,kessels_3state_stab_proof,dijkstra_3state_stab_proof,dijkstra_4state_stab_note,mechanized_k_state_stab,ghosh_stab,regular_abstractions,rewrite_self_stab,automatic_convergence_self_stab,self_stab_coq}.

\paragraph{Contributions.}
\begin{itemize}
    \item We define \emph{implicit rankings}, which soundly approximate reduction and conservation of some ranking function in FOL, and show how these implicit rankings can be used in liveness proofs (\Cref{sec:ranking}).
        \item We introduce constructors of implicit rankings based on familiar notions from order theory,
        including domain-based constructors that are sound for finite but unbounded domains and can sometimes replace reasoning about unbounded summations and set cardinalities (\Cref{sec:constructions}).
    \item We implement the proposed constructors in a tool for verifying liveness properties of first-order transition systems, 
    and demonstrate the power of our constructors by verifying a set of examples of liveness properties, including Dijkstra's self-stabilizing protocols (\Cref{sec:evaluation}).
    \end{itemize}
\ifshort
In the next section we present two examples that motivate our approach. \Cref{sec:prelims} then provides the necessary background for  \Crefrange{sec:ranking}{sec:evaluation} and
\Cref{sec:related} discusses related work and concludes the paper. 
\else 
The rest of the paper is organized as follows. \Cref{sec:motivating_examples} describes two motivating examples. \Cref{sec:prelims} provides the necessary background on transition systems in FOL and order theory. %
\Cref{sec:ranking,sec:constructions,sec:evaluation} describe the aforementioned contributions.
\Cref{sec:related} discusses related work and concludes the paper. 
In \Cref{appendix:abstractions,appendix:definition,appendix:heights,appendix:linsums,appendix:proof_rule,appendix:proofs,appendix:totality} we discuss some more interesting aspects of our constructors, give proofs of their soundness and expand on our evaluation.
\fi

%% file: sections/examples.tex
\vspace{-0.2cm}
\section{Motivating Examples}\label{sec:motivating_examples}
\begin{example}\label{example:toystab}
As a first motivating example (\Cref{fig:toy_stab}), we consider an abstraction of a self-stabilizing protocol.
In a self-stabilizing protocol privileges are initially distributed arbitrarily in a network, but eventually the protocol converges to a stable state where only one machine holds a privilege.
The protocol we consider abstracts the movement of privileges in Dijkstra's $k$-state protocol~\cite{dijkstra_self_stab}.
While the abstraction does not enjoy stabilization, it suffices for showing a property %
that is used in the proof of stabilization for Dijkstra's protocol.

In this protocol, the network consists of a finite number of machines, arranged in a ring, with one machine, called the bottom machine ($\botm$), being distinguished.
Every machine $i$ has a field $i.\nextf$ which directs to its right-hand neighbor in the ring.
At any moment any machine may be privileged or not. 
Initially, privileges are assigned arbitrarily, with the guarantee that at least one machine is privileged.
At each iteration, an arbitrary privileged machine is scheduled, the scheduled machine loses its privilege and its right-hand neighbor becomes privileged (whether it was already privileged or not).
The property we wish to prove is that machine $\botm$ is eventually scheduled: $\eventually (\skd = \botm)$.
In the original protocol, the steps of machine $\botm$ are different from all other machines (this difference is lost in the abstraction) and, in particular, they take the state of the network closer to stabilization.

\ifshort\else
\vspace{-0.6cm}
\input{figures/stabilization_code_example}
\vspace{-0.5cm}
\fi

Next, we describe a ranking function for proving the liveness property.
\ifshort For this,\else To present the ranking function \fi we denote the number of machines by $n$ and think of the machines as indexed by $0,1,\ldots,n-1$ according to their order in the ring, with $\botm=0$.
Now, a natural ranking function is:  $\rankname=\sum_{i\neq 0 : \priv(i)} n-i$.
Intuitively, the value $n-i$ can be seen as the number of steps required for the privilege of machine $i$ to reach machine $0$ (if no other privileges are present).
We can now verify that in any transition from $s$ to $\struct'$, if the eventuality does not occur in either $\struct$ or $\struct'$, the rank is reduced.
Indeed, for any transition, if $\skd = 0$ the eventuality is satisfied in $\struct$;
otherwise, $\skd$ loses its privilege, decreasing $\rankname$ by $n-\skd$, and
$\skd+1$ gains a privilege, (possibly) increasing $\rankname$ by $n-(\skd+1)$, ultimately decreasing $\rankname$ by $1$ (or more, if $\skd+1$ was already privileged).

\ifshort
\input{figures/both_example_code}

\else\fi

While this ranking function is natural, it is not clear how to reason about it with existing automated solvers  due to the unbounded sum operation which is not directly expressible in FOL (without induction).
Next, we show how we can express this ranking argument in FOL by an implicit ranking which encodes a sound approximation of the  reduction in the ranking function. To do this, we show that we can replace reasoning about set cardinalities and unbounded sums by reasoning about other aggregations that are expressible in FOL.

To avoid combining arithmetic reasoning with quantifiers, we use an uninterpreted sort to model the machines, and instead of relying on numbers, we model the ring structure with a strict order on machines, $\lt$, such that $\botm$ is minimal in that order, and the next field corresponds to the successor function in the order $\lt$, except for the maximal element which points to bot.%

Now, we observe that in our model, the expression $n-i$ used in the ranking function above for a privileged machine $i$ is the cardinality of the set of machines that are greater or equal to $i$ in the order $\lt$. 
Further, we can sum over all machines by pushing the condition on $i$ into the set definition. Thus, 
$\rankname= 
\sum_i |\{j : \priv(i) \wedge i\neq 0 \wedge (\lt(i,j)\vee i=j)\}|$.
In order to express reduction in the sum, %
we recall that in a transition, the summands are unchanged for all machines except for 
$\skd$, $\skd+1$; 
the contribution of $\skd+1$ to the sum after the transition is smaller than the contribution of $\skd$ before the transition and the contribution of $\skd$ to the sum after the transition is $0$ (because it loses its privilege) and therefore, smaller or equal to that of $\skd+1$ before the transition.
Therefore we can use an adaption of a pointwise argument which ``permutes'' $\skd$ and $\skd+1$ to capture the reduction in the sum over all machines.

Next, we need to express reduction (or conservation) in the summands (either of the same machine $i$ or of different machines) along transitions. The summands are cardinalities of sets of machines $j$. Thus, we observe that it suffices to show strict set inclusion.
This amounts to a pointwise argument over all machines $j$, showing reduction \ifshort\else(or conservation) \fi in the binary predicate $\formula(i,j)=\priv(i) \wedge i\neq 0 \wedge (\lt(i,j)\vee i=j)$that defines membership of $j$ in the set of $i$.
The reduction in the binary predicate, the pointwise reduction and the permuted-pointwise reduction, can all be encoded in FOL.
Thus, the overall approximation of the reduction in the ranking function can be expressed (and hence verified automatically) by the following first-order formula, where unprimed symbols represent the pre-state of a transition and primed variables represent the post-state: 
\ifshort\vspace{-0.1cm}\else\fi
\begin{align*}
    &\forall j.\ (\alpha'(\nextf(\skd),j)\to \alpha(\skd,j))\wedge 
    \exists j.\ ( \neg\alpha'(\nextf(\skd),j)\wedge   \alpha(\skd,j))\ \wedge  \\ 
    &\forall j. \ (\alpha'(\skd,j)\to \alpha(\nextf(\skd),j)) \
    \wedge \\
    & \forall i. \ ((i\neq \skd \wedge i\neq \nextf({\skd}) \to \forall j. \ (\alpha'(i,j)\to \alpha(i,j)) )
\end{align*}
\end{example}
\ifshort\vspace{-0.25cm}\else\fi

\begin{example}\label{example:binary_counter}
 As a second motivating example (\Cref{fig:binary_counter}), we look at a binary counter  implemented in an array, taken from \cite{runtime_squeezers}. For simplicity of the presentation, we consider a version of the counter that counts down.
The counter is implemented by an array $\arr[0],\ldots,\arr[n-1]$ with $\arr[0]$ as the most significant bit and $\arr[n-1]$ as the least significant bit, initialized with all $1$s. 
A pointer $\ptr$ traverses the array starting at index $n-1$ (the lsb) until it sees the first $1$, replacing at each step any $0$ it sees with a $1$. 
When reaching the first cell with a $1$, it sets that cell to a $0$ and returns $\ptr$ to index $n-1$. %
We wish to prove that eventually the array holds all $0$s: $\eventually (\forall x. \arr[x]=0)$.

\ifshort\else
\vspace{-0.6cm}
\input{figures/binary_code_example}
\vspace{-0.3cm}
\fi

To find a ranking function, we notice that the reduction in the counter value happens  between states of the program where $\ptr=n-1$. Such states partition the execution to ``intervals''. For a state inside an interval, we can derive the value of the counter at the beginning of the interval (i.e.,  when $\ptr$ was last $n-1$) by 
$\arr_{\mathsf{old}}[i]=0 \iff \arr[i]=0 \vee i > \ptr$.
Within an interval, 
the counter values stored by $\arr_{\mathsf{old}}$ do not change but $\ptr$ is reduced; and when a new interval starts, the value of $\arr_{\mathsf{old}}$ is reduced (and $\ptr$ is set back to $n-1$). 
This lends itself to a ranking function in the form of a lexicographic pair 
$\rankname = (\mathrm{val}(\arr_{\mathsf{old}}), \ptr)$. %

Unfortunately, it is not clear how to encode the value of the counter in FOL. In fact, most existing techniques for encoding ranking functions in FOL are limited to polynomial ranking functions, while the counter requires an exponential ranking function.
Fortunately, we realize that 
a reduction in the counter value corresponds to a lexicographic reduction in the sequence of bits stored in the array representing the counter based on the order on array indices.%

To encode the ranking argument in FOL we model the indices of the array, as in \Cref{example:toystab}, by an uninterpreted index sort and a strict order $\lt$ on it, 
with the maximal index acting as $n-1$.
The content of the array is modeled using a unary relation $\arr(\cdot)$ that records the cells with value $1$. Then, $\arr_{\mathsf{old}}$ is recorded by $\arr_{\mathsf{old}}(i) =\arr(i) \wedge (\lt(i,\ptr)\vee i=\ptr)$.
With this encoding, reduction in $\ptr$ is measured by the order on indices $\lt$ and reduction in $\arr_{\mathsf{old}}$ is measured lexicographically, based on the same order $\lt$ on indices, where for every index $i$, the values of cell $i$ are ordered by implication on the derived predicate $\arr_{\mathsf{old}}(i)$.
Overall, the reduction in rank can be encoded in FOL by the formula: 
\ifshort\vspace{-0.1cm}\else\fi
\begin{align*}
    & \exists i. \ ( \neg \arr_{\mathsf{old}}'(i) \wedge  \arr_{\mathsf{old}}(i)
    \wedge \forall j. \ \lt(j,i) \to ( \arr_{\mathsf{old}}'(j) \to  \arr_{\mathsf{old}}(j))) \ \vee\\
    & ((\forall i. \  \arr'_{\mathsf{old}}(i) \leftrightarrow  \arr_{\mathsf{old}}(i)) \wedge \lt(\ptr',\ptr)) 
\end{align*}\end{example}
\Cref{sec:ranking}  formally defines the notion of an implicit ranking which the  reduction formulas above exemplify.
The crux of our work is in \Cref{sec:constructions}, where we present  constructors of implicit rankings. As we show in \Cref{sec:evaluation}, the implicit rankings for both motivating examples can be built using our constructors.

%% file: figures/stabilization_code_example.tex
\begin{figure}
\centering
\lstset{style=mystyle}
\begin{minipage}{0.4\textwidth}
\begin{lstlisting}[language=custom]
machine bot
relation priv(machine)
function next(machine) : machine
while(exists m. priv(m)):
    skd = *
    assume(priv(skd))
    priv(skd) = false
    priv(next(skd)) = true
\end{lstlisting}
\vspace{-0.25cm}
\end{minipage}
\caption{Toy Stabilization \label{fig:toy_stab}}
\end{figure}

%% file: figures/both_example_code.tex
\begin{figure}[t]
    \centering
    \begin{minipage}{0.45\textwidth}
        \centering
\lstset{style=mystyle,   showlines=true
}
\begin{lstlisting}[language=custom]
machine bot
relation priv(machine)
function next(machine) : machine
while(exists m. priv(m)):
    skd = *
    assume(priv(skd))
    priv(skd) = false
    priv(next(skd)) = true
            
\end{lstlisting}
\caption{Toy Stabilization \label{fig:toy_stab}}
    \end{minipage}\hfill
    \begin{minipage}{0.45\textwidth}
        \centering
\lstset{style=mystyle}
\begin{lstlisting}[language=custom]
index ptr = n - 1
array(index) a = [1 for i in 0,...,n-1]
while (ptr >= 0):
    if a[ptr] == 0:
        a[ptr] = 1
        ptr = ptr - 1
    else: 
        a[ptr] = 0
        ptr = n - 1
\end{lstlisting}
        \caption{Binary Counter \label{fig:binary_counter}}
    \end{minipage}
\vspace{-0.3cm}
\end{figure}

%% file: figures/binary_code_example.tex
\begin{figure}
\centering %
\lstset{style=mystyle}
\begin{minipage}{0.45\textwidth} %
\begin{lstlisting}[language=custom]
index ptr = n - 1
array(index) a = [1 for i in 0,...,n-1]
while (ptr >= 0):
    if a[ptr] == 0:
        a[ptr] = 1
        ptr = ptr - 1
    else: 
        a[ptr] = 0
        ptr = n - 1
\end{lstlisting}
\vspace{-0.25cm}
\end{minipage}
\caption{Toy Stabilization \label{fig:binary_counter}}
\end{figure}

%% file: sections/prelims.tex
\ifshort\vspace{-0.2cm}\else\fi
\section{Preliminaries} \label{sec:prelims}
\paragraph{First-Order Logic.} We consider uninterpreted FOL with equality, without theories.
For simplicity, we present our results for a single-sorted version of FOL.
The extension of the results to many-sorted logic, which we use in practice, is natural.
A first-order signature $\Sigma$ contains relation, constant and function symbols. A term $t$ over $\Sigma$ is a variable $x$, a constant $c$, the application of a function on a sequence of terms $f(\seq t)$ or an if-then-else term $\mathrm{ite}(\formula,t_1,t_2)$.
Formulas $\formula$ over $\Sigma$ are defined recursively: atomic formulas are either $t_1=t_2$ or $r(\seq t)$ where $r$ is a relation symbol.
Non-atomic formulas are built using connectives $\neg,\wedge,\vee,\to$ and quantifiers $\forall,\exists$.
We write  $\formula(\seq {x_1},\ldots,\seq {x_k})$ to denote that the free variables in $\formula$ are contained in $\seq{x_1}, \ldots,\seq {x_k}$.
Given sequences of terms $\seq {t_1},\ldots, \seq {t_k}$ we then use the notation $\formula(\seq {t_1},\ldots,\seq {t_k})$ to denote the formula obtained from $\formula$ by substituting each $\seq {x_j}$ with $\seq {t_j}$.
For a signature $\Sigma$ we use subscripts and superscripts $\ddag$ to denote disjoint copies of $\Sigma$ defined by $\Sigma\ddag = \{a\ddag \mid a\in \Sigma\}$, assumed to satisfy $\Sigma\cap \Sigma\ddag = \emptyset$. 
For a formula $\formula$ over $\Sigma$, we denote by $\formula\ddag$ the formula over $\Sigma\ddag$ obtained by substituting each symbol $a\in \Sigma$
 with $a\ddag\in \Sigma\ddag$.

First-order formulas are evaluated over pairs of structures and assignments. 
A structure for $\Sigma$ is a pair $\struct =(\domain,\interp)$ where $\domain$ is a non-empty set called the domain, and $\interp$ is an interpretation that maps each relation, constant and function symbol to an appropriate construct over $\domain$.
We denote by $\structset(\Sigma,\domain)$ the set of all structures over $\domain$.
An assignment $\assign$ from a sequence of variables $\seq x$ to a domain $\domain$ is a function $\assign: \seq x \to \domain$.
We denote the set of all assignments from $\seq x$ to $\domain$ by $\assignset(\seq x,\domain)$.
We denote the concatenation of two sequences of variables $\seq x$, $\seq y$ by $\seq x\concatvar \seq y$.
For two assignments $\otherassign,\assign$ without shared variables, we denote by $\otherassign\concatfunc\assign$ the disjoint union of $\otherassign$ and $\assign$.
We call a pair $(\struct,v)$ of a structure and an assignment an \pairname. For a domain $\domain$ and a sequence of variables $\seq x$ we denote the set of all {\pairnames} 
by 
$\pairset(\Sigma,\seq x, \domain)$.
For a formula $\formula$ and an \pairname \ $(\struct,\assign)$ we write $(\struct,\assign)\models \formula$ to denote that $(\struct,\assign)$ satisfies $\formula$.

\paragraph{Transition Systems in First-Order Logic.}
We consider transition systems given by a first-order specification $\Tspec = (\Sigma,\iota,\tau)$ where $\Sigma$ is a signature, %
$\iota$ is a closed formula over $\Sigma$ that specifies initial states, and $\tau$ is 
a formula over a double signature $\Sigma\uplus \Sigma'$ that specifies transitions, where the symbols in $\Sigma$ represent the pre-state and the symbols in $\Sigma'$ represent the post-state of a transition. 
A trace of $\Tspec$ is an infinite sequence of structures  $(\struct_i)_{i=0}^\infty$ over the same domain
such that $\struct_0$ is an initial state and for all $i\in \nat$, there is a transition from $\struct_i$ to $\struct_{i+1}$.

\paragraph{Well-Founded Partial Orders.}
A binary relation $\leq$ on a set $A$ is called a partial order if it is reflexive, antisymmetric and transitive.
For a partial order we write $a_1 < a_2$ as shorthand for $a_1\leq a_2$ and $a_1 \neq a_2$.
A partial order $\leq$ is called well-founded if there is no sequence $a_0,a_1,\ldots$ such that $a_{i+1} < a_i$ for all $i\in\nat$. We then refer to  
$\leq$ as a wfpo for short\footnote{This is a different (weaker) notion from a partial-well-order (an antisymmetric well-quasi-order), which requires no infinite decreasing chains and no infinite antichains.}. 
If $\leq$ is a partial order on $A$ and $A$ is finite, then $\leq$ is a wfpo.

%% file: sections/ranks.tex
\section{Expressing Ranking in First-Order Logic}\label{sec:ranking}

In this section we introduce the notion of \emph{implicit ranking}.
To motivate our definition, we start by considering the way rankings are typically used.

\ifshort\vspace{-0.2cm}\else\fi
\subsection{Using Ranking for Liveness Proofs}\label{subsec:using_ranking}

Well-founded rankings are useful for proving liveness properties of transition systems.
A typical proof rule that uses rankings is based on two notions: a conservation of the ranking, corresponding to $\leq$, and a reduction in the ranking, corresponding to $<$.
As an example, we examine proving liveness properties of the form
$\liveprop = \globally\eventually r \to \eventually q$. A transition system $\Tspec=(\Sigma,\iota,\tau)$ satisfies $\liveprop$ if every trace of $\Tspec$ that satisfies $r$ infinitely often eventually satisfies $q$ ($r$  can be understood as a fairness assumption). 
We can prove such a property by finding a ranking function $f$ from states of the system to a set $A$ with a wfpo $\leq$ and a formula $\triggerinv$, and validating the following premises:
\vspace{-0.2cm}
\begin{multicols}{2}
\begin{enumerate}
    \item $\iota\wedge \neg q \to \triggerinv$ \label{item:init}
    \item $\triggerinv \wedge  \tau \wedge \neg q' \to \triggerinv' $ \label{item:consec}
    \item $ \triggerinv \wedge \tau \wedge \neg q' \to \conserved$ \label{item:conserved}
    \item $\triggerinv \wedge \tau \wedge \neg q' \wedge r  \to \reduced$ \label{item:reduced}%
\end{enumerate}
\end{multicols}
\vspace{-0.2cm}
\noindent
where $\conserved,\reduced$ are formulas that encode conservation and reduction of $f$.
\Cref{item:init,item:consec} assert that $\triggerinv$ holds in all states in a trace as long as $q$ does not hold.
\Cref{item:conserved} guarantees that in every transition that does not visit $q$ we have a conservation of $f$. 
\Cref{item:reduced} states that in every transition following a visit to $r$ we have a reduction in $f$.
If all premises are valid, $r$ cannot be visited infinitely often without eventually satisfying $q$, as that would induce an infinitely decreasing chain in $f$\ifshort. \else, and so, $q$ must be reached after some finite number of visits to $r$. \fi
In \Cref{appendix:proof_rule}, we show a proof rule for more general liveness properties that we use in our evaluation, with similar premises to the rule above.

The structure of the proof rule above reveals that for soundness, $\conserved$ and $\reduced$ do not need to precisely capture conservation and reduction in the ranking. 
Instead, because $\conserved,\reduced$ appear only positively, it suffices that they \emph{underapproximate} them, such that whenever a pair of states satisfies $\conserved$, the value of $f$ is conserved between them, and whenever a pair of states satisfies $\reduced$, the value of $f$ is reduced between them.
The implication in the other direction is not needed for soundness.
This observation is key for encoding complex ranking arguments in FOL.
It allows flexibility in encoding conservation and reduction, which is crucial in cases where these relations (as well as the ranking function itself) are not directly expressible in FOL, but their underapproximations are.

\ifshort\vspace{-0.2cm}\else\fi
\subsection{Implicit Ranking}

\ifshort
This
\else
Motivated by the observation that it suffices to encode underapproximations of conservation and reduction in the ranking
this
\fi 
section formalizes the 
notion of an \emph{implicit ranking}\ifshort --- a 
\else. An implicit ranking is a
\fi
pair of formulas $\conserved,\reduced$ which, for every domain $\domain$, encode underapproximations of conservation and reduction of some implicitly defined ranking function on structures with domain $\domain$. For compositionality, we generalize this concept to pairs of {\pairnames}.

\ifshort \vspace{-0.1cm} \else\fi
\paragraph{Notation.}
The definition of an implicit ranking uses formulas that reason about a pair of {\pairnames}.
To that end, given a signature $\Sigma$ and a sequence of variables $\seq x = ({x_{i}})_{i=1}^m$, we consider a double signature $\Sigma_\low\uplus\Sigma_\high$
and two copies of the variables 
$\seq x_b = ({x_{b,i}})_{i=1}^m$ for $b \in \{\low, \high\}$. 
Intuitively, $\Sigma_\low$ and $\seq x_\low$ represent a ``lower ranked'' {\pairname} and $\Sigma_\high$ and $\seq x_\high$ represent a ``higher ranked'' {\pairname}.
For a term $t$ with variables $\seq x$, we denote by $t_b$ the term obtained by substituting $\Sigma_b$ for $\Sigma$ and $\seq{x_b}$ for $\seq {x}$.
With abuse of notation, we 
consider a structure $\struct = (\domain,\interp)$ for $\Sigma$ as a structure for $\Sigma_b$  where $\interp(a_b) = \interp(a)$ for every $a_b \in \Sigma_b$ and an assignment to $\seq x$ as an assignment to $\seq{x_b}$ by defining $\assign(x_{b,i}) = \assign(x_i)$.
For a formula $\varphi$ over $\Sigma_\low \uplus \Sigma_\high$ with free variables $\seq {x_\low},\seq {x_\high}$ and 
{\pairnames} $(\struct_\low,\assign_\low),(\struct_\high,\assign_\high)\in \pairset(\Sigma,\seq x,\domain)$, we write $(\struct_\low,\assign_\low),(\struct_\high,\assign_\high)\models \varphi$ to denote satisfaction of $\varphi$ when for $b \in \{0,1\}$, the interpretation of $\Sigma_b$ is taken from $\struct_b$, and the assignment to  $\seq {x_b}$ is taken from  $\assign_b$.

\begin{definition}[Implicit Ranking]\label{def:implicit_ranking}
    Let $\conserved,{\reduced}$ be two formulas over  $\Sigma_\low\uplus\Sigma_\high$. We say that $\rankname=({\conserved},{\reduced})$ 
    is an \emph{implicit ranking} with \parameters \ $\seq x$,
    if
    the free variables in $\conserved$ and $\reduced$ are confined to ${\seq x_\low},\seq {x_\high}$
    and 
    for any domain $\domain$, there exist a set $A$, a wfpo $\leq$ on $A$ and a function $f: \pairset(\Sigma,\seq x,\domain) \to A$, such that:
    \begin{align*}
        &(\struct_\low,\assign_\low),(\struct_\high,\assign_\high)\models \conserved(\seq {x_\low},\seq {x_\post}) \implies f(\struct_\low,\assign_\low) \leq f(\struct_\high,\assign_\high)\\
        &(\struct_\low,\assign_\low),(\struct_\high,\assign_\high)\models \reduced(\seq {x_\low},\seq {x_\post}) \implies f(\struct_\low,\assign_\low) < f(\struct_\high,\assign_\high)
    \end{align*}
    For an implicit ranking $\rankname$ and a domain $\domain$, we call 
    $(A,\leq)$ a \emph{ranking range} for $\domain$ and $f$ a \emph{ranking function} for $\domain$.
    If $\seq x$ is empty, we call $\rankname$ a \emph{closed} implicit ranking.
    If the existence of $A,\leq$ and $f$ is ensured only for finite domains $\domain$, we call $\rankname$ an \emph{implicit ranking for finite domains}.
\end{definition}
The formulas $\conserved,\reduced$ encode  relations between {\pairname}s, which means that they let us compare specific elements in one state with possibly different elements in  another state. Technically, this is achieved by the free variables in $\conserved$ and $\reduced$. 
In proof rules, \ifshort \else such as the one above, \fi we only need to compare structures and \ifshort thus \else as such we \fi only allow the use of closed implicit rankings.
The use of free variables is needed  for the construction of these rankings. 
In \Cref{sec:constructions} we introduce constructors that create new implicit rankings from existing rankings, some of which 
eliminate free variables.
\ifshort\else Our use of parameters in implicit rankings is analogous to the fact that FOL formulas may have free variables but in a proof we use only closed formulas (sentences).\fi

\Cref{def:implicit_ranking} ensures that for each domain $\domain$, the formulas $\conserved$ and $\reduced$ underapproximate conservation and reduction of 
some ranking function over {\pairnames}.
The ranking function $f$ and the corresponding ranking range are not explicitly encoded. 
As explained above, this allows us to find such formulas even if the exact conservation and reduction of $f$ are not readily expressible in FOL.
For example, for a function $f$ that maps $\struct$ to the number of elements in $\domain$ that satisfy a predicate $\alpha(x)$, there is no formula that precisely captures the conservation of $f$, but it can be approximated, say, by the formula $\forall x. \alpha_\low(x)\to \alpha_\high(x)$.

Closed implicit rankings can be used in proof rules such as the one above by substituting $\Sigma'$ (the post-state signature) for $\Sigma_\low$ and $ \Sigma$ (the pre-state signature) for $\Sigma_\high$.
While implicit rankings ensure reduction/conservation in a possibly different function for each domain, they are sound to use as all states along a trace of a transition system share a domain.
In the case of implicit rankings for finite domains, soundness is guaranteed for systems where the domain is finite in each trace, but still unbounded.
In \Cref{appendix:definition} we discuss alternative definitions of implicit rankings and compare them to \Cref{def:implicit_ranking}.

%% file: sections/constructions.tex
\section{Constructions of Implicit Rankings}\label{sec:constructions}

We now introduce several constructors for implicit rankings, which can be instantiated and composed to create a rich family of implicit rankings. 
Each constructor has its own arguments, which may themselves be implicit rankings.
Some of the constructors only construct implicit rankings for finite domains. These are called finite-domain constructors.
For other constructors, whether the constructed implicit ranking is for arbitrary or finite domains is inherited from the implicit rankings used as arguments for the constructor. 
We provide the following guarantees.
\begin{theorem}
All constructors defined in this section are \emph{sound} in the sense that (i)~if the arguments of a constructor %
satisfy their assumptions, then the output of the constructor is an implicit ranking (for finite domains if the constructor is finite-domain), 
and (ii)~for constructors that receive implicit rankings as arguments, if at least one of the arguments is an implicit ranking for finite domains, so is the constructed implicit ranking.
\end{theorem}
Our constructors essentially encode standard constructions of partial orders (or their under-approximations) in FOL, accompanied by suitable definitions of ranking functions. \ifshort Due to space constraints we defer \else We defer \fi the definition of two constructors which are less central to our work to \Cref{appendix:linsums}. 
We further defer 
the soundness proofs to \Cref{appendix:proofs}\ifshort.\else \ and give only the crux of the proofs, which is the definition of ranking range and a ranking function for the constructed implicit ranking given corresponding notions for the arguments of the constructor. \fi 
In \Cref{appendix:heights} we present intuitive (but less general) constructions of ranking functions for our constructors.

\subsection{Base Constructors}\label{subsec:base_constructors}

We start by introducing two non-recursive constructors, which are used as the base in recursive constructions of implicit rankings. 

\paragraph{\bf Binary Constructor.}
The first constructor we define 
captures \ifshort\else the simplest possible ranking function --- \fi a binary ranking function, mapping each \pairname \ $(\struct,\assign)$ to $0$ or $1$ (ordered $0 \leq 1$) by checking whether $(\struct,\assign)$ satisfies a formula $\formula$. In this case, reduction \ifshort\else in the ranking \fi is obtained if $\formula$ holds in the higher-ranked \pairname \ and does not hold in the lower-ranked \pairname.
Conservation also includes the case where $\formula$ holds in both or in neither.

\begin{constructor}\label{const:Bin}
The \emph{binary constructor} \ifshort receives \else receives as input \fi a formula $\formula(\seq x)$ over $\Sigma$. It returns an implicit ranking $\mathrm{Bin}(\formula)=(\conserved,\reduced)$ with parameters $\seq x$ defined by:\begin{center}
    $ \conserved(\seq {x_\low},\seq {x_\high}) = \formula_\low(\seq {x_\low}) \to \formula_\high(\seq {x_\high}) \qquad
    {\reduced}(\seq {x_\low},\seq {x_\high}) = \formula_\high(\seq {x_\high}) \wedge \neg \formula_\low(\seq {x_\low})
    $
\end{center}
\end{constructor}

\begin{example}\label{ex:Bin}
    Continuing with \Cref{example:toystab}, for 
    $\formula(i,j)=\mathrm{priv}(i) \wedge \mathrm{lt}(i,j)$, 
    $\mathrm{Bin} (\formula)$
    is an implicit ranking \ifshort \else $(\conserved,\reduced)$ \fi with parameters $i,j$\ifshort, capturing \else.  
    The formulas $\conserved,\reduced$ capture \fi  reduction and conservation in the ranking between a pair of machines $(i_\high,j_\high)$ in the higher-ranked structure and a potentially different pair $(i_\low,j_\low)$ in the lower-ranked structure.
    The need for comparing different pairs in different structures is demonstrated in \Cref{example:toystab}.   
\end{example}

\paragraph{\bf Position-in-Order Constructor.}
The second base constructor, which is a finite-domain constructor, utilizes an already existing (possibly derived) partial order
in the system itself, 
building on the observation that a partial order over a \emph{finite} domain is always well-founded. 
The order is defined by a single signature formula $\orderformula(\seq {y_\low},\seq {y_\high})$, which allows comparing two (tuples of) elements. %
To guarantee soundness we must verify that $\orderformula(\seq {y_\low},\seq {y_\high})$ defines a strict order, and that it is immutable, i.e., does not change between the ranked structures\ifshort , encoded by: \else. These requirements are encoded by the closed formula:\fi
\begin{align*}
&\immutorder(\orderformula) := \ 
(\forall \seq {y^1},\seq {y^2}. \orderformula_\low(\seq {y^1},\seq {y^2})\leftrightarrow \orderformula_\high (\seq {y^1},\seq {y^2})) \ \wedge\\
&\quad (\forall \seq {y^1},\seq {y^2}. \orderformula_\low(\seq {y^1},\seq {y^2})\to \neg \orderformula_\low(\seq {y^2},\seq {y^1})) \wedge
(\forall \seq {y^1},\seq {y^2},\seq {y^3} . \orderformula_\low(\seq {y^1},\seq {y^2}) \wedge \orderformula_\low(\seq  {y^2},\seq {y^3}) \to \orderformula_\low(\seq {y^1}, \seq {y^3}))
\end{align*}
If $\struct_\low,\struct_\high 
\models \immutorder(\orderformula)$ the interpretation of $\orderformula$ in $\struct_\low$ coincides with its interpretation in $\struct_\high$ and both define a strict partial order on the set $\assignset(\seq y,\domain)$, where $\domain$ is the domain of $\struct_\low$ and $\struct_\high$.
Then, if $\domain$ is finite, this order is a wfpo.
The position-in-order constructor uses $\orderformula$ to create an implicit ranking that compares the valuation of the same term (or sequence of terms) $\seq t$ in the two structures. The term may include free variables, allowing to compare valuations under 
different assignments in the two structures.
\ifshort
The underlying ranking function maps an {\pairname} to its interpretation of $\orderformula$ together with the valuation it assigns to $\seq t$.
\else\fi

\begin{constructor}\label{const:Pos}
The \emph{position-in-order constructor} \ifshort receives \else receives as input \fi a formula $\orderformula(\seq {y_\low},\seq {y_\high})$ over $\Sigma$ and a sequence of terms $\seq t(\seq y)$ over $\Sigma$ with $|\seq t|=|\seq y|$.
It returns an implicit ranking for finite domains $\mathrm{Pos}(\seq t,\orderformula)=(\conserved,\reduced)$ with parameters $\seq y$ defined by:
$$\textstyle
{\conserved}(\seq {y_\low} ,\seq {y_\high}) = \immutorder(\orderformula) \ \wedge \ (\orderformula_\low(\seq {t_\low},\seq {t_\high}) 
\ \vee \ \seq {t_\low} = \seq {t_\high}) \qquad
    {\reduced}(\seq {y_\low} ,\seq {y_\high}) = \immutorder(\orderformula) \ \wedge \  \orderformula_\low(\seq {t_\low},\seq {t_\high})
$$
\end{constructor}
\ifshort 
\else
The definition of a ranking function  that justifies soundness of the constructed implicit ranking for a given domain $\domain$ is somewhat subtle since the order induced by $\orderformula$ depends on the specific structure considered and not just its domain.
Thus, we include the interpretation of $\orderformula$ in the ranking range.
That is, we define $A = \mathrm{interp}(\orderformula,\domain) \times \assignset(\seq y,\domain)$, where $\mathrm{interp}(\orderformula,\domain)$ denotes the set of all possible interpretations of $\orderformula$ in structures with domain $\domain$.
An interpretation $\interp^\orderformula$ of $\orderformula$ on $\domain$ is a subset of $\assignset(\seq {y_\low} \concatvar \seq {y_\high},\domain)$. 
The order $\leq$ on $A$ is defined by: $(\interp_\low^\orderformula,\assign_\low) \leq (\interp_\high^\orderformula,\assign_\high)$ if $\interp_\low^\orderformula=\interp_\high^\orderformula$ and $\assign_\low \concatfunc \assign_\high \in \interp_\low^\orderformula$ or $\assign_\low = \assign_\high$.
For a structure $\struct$, the interpretation $\interp^\orderformula(\struct)$ of $\orderformula$ is the set of assignments $\assign_\low \concatfunc \assign_\high \in  \assignset(\seq {y_\low} \concatvar \seq {y_\high},\domain)$ such that $\struct,\assign_\low \concatfunc \assign_\high \models \orderformula$.
A ranking function to $(A,\leq)$ is then defined by $f(\struct,\assign)=(\interp^\orderformula(\struct),\assign^{\seq t})$ where $\assign^{\seq t}$ is an assignment to $\seq y$ that takes its value based on the values $\assign$ gives to $\seq t$.
\fi

\begin{example}\label{ex:Pos}
    Continuing with \Cref{example:binary_counter}, define $\orderformula(i_\low,i_\high) = \mathrm{lt}(i_\low,i_\high)$, and $t = \mathrm{ptr}$. Then $\mathrm{Pos}(t,\orderformula)$ is a closed implicit ranking for finite domains, where the   ranking is based on the position of the pointer in the different structures.
\end{example}

\subsection{Constructors for Aggregation of Finitely Many Rankings}\label{subsec:finite_aggregations}

We now present constructors which receive as input a finite number of implicit rankings and create an implicit ranking that mimics aggregation of the underlying ranking ranges and functions. 
These constructors are based on standard ways to lift partial orders on sets to a partial order on their Cartesian product.

\paragraph{\bf Pointwise Constructor.}
For partially-ordered sets $A_1,\ldots,A_m$ with partial orders $\leq_1,\ldots,\leq_m$ respectively, the pointwise partial order 
$\leq_{\text{pw}}$ on the set $A_1\times\cdots\times A_m$ is defined by 
$(a_1,\ldots,a_m)\leq_{\text{pw}} (b_1,\ldots,b_m) \iff \bigwedge_i a_i\leq_i b_i$.
If the orders $\leq_1,\ldots,\leq_m$ are wfpos then so is $\leq_{\text{pw}}$.
The following constructor 
uses a FOL encoding of a pointwise partial order to aggregate rankings. 
\ifshort A corresponding ranking function to $A_1 \times \cdots \times A_m$ is defined by $f(\struct,\assign) = (f_1(\struct,\assign),\ldots,f_m(\struct,\assign))$, where $f_i$ is a ranking function to $A_i$.\else\fi

\begin{constructor} \label{const:PW} 
The \emph{pointwise constructor} \ifshort receives \else receives as input \fi
implicit rankings $\rankname^i=(\conservedsuper{i},\reducedsuper{i})$ 
for $i=1,\ldots,m$, each with parameters $\seq x$.
It returns an implicit ranking $\mathrm{PW}(\rankname^1,\ldots,\rankname^m) = ({\conserved},{\reduced})$ with parameters $\seq x$ defined by:
\begin{center}
    $
{\conserved}(\seq {x_\low},\seq {x_\high}) = \bigwedge_i \conservedsuper{i}(\seq {x_\low},\seq {x_\high})
\qquad
{\reduced}(\seq {x_\low},\seq {x_\high}) = {\conserved}(\seq {x_\low},\seq {x_\high}) \wedge \bigvee_i \reducedsuper{i}(\seq {x_\low},\seq {x_\high})
$
\end{center}

\end{constructor}
\ifshort\else
Given implicit rankings $\rankname^1,\ldots,\rankname^m$,
soundness of $\rankname = \mathrm{PW}(\rankname^1,\ldots,\rankname^m)$ holds since for every domain $\domain$, 
given a ranking range $(A_i,\leq_i)$ and a ranking function $f_i$ of each $\rankname_i$, 
a ranking range of $\rankname$ for $\domain$ can be defined by $(A_1\times \cdots \times A_m, \leq_{\text{pw}})$
and a ranking function can be defined by $f(\struct,\assign) = (f_1(\struct,\assign),\ldots,f_m(\struct,\assign))$.\fi

\paragraph{\bf Lexicographic Constructor.}
A different partial order on $A=A_1\times \cdots\times A_m$ can be defined by the lexicographic ordering: 
$(a_1,\ldots,a_m)\leq_{\text{lex}} (b_1,\ldots,b_m) \iff \bigvee_i (a_i <_i b_i \wedge \bigwedge_{j<i} (a_j \leq_j b_j)) \vee (\bigwedge_i a_i \leq_i b_i)$.
Again, if $\leq_i$ are all wfpos, then so is $\leq_{\text{lex}}$. 
The Lexicographic Constructor encodes this idea in FOL.
\ifshort
A corresponding ranking function is defined as for the pointwise constructor.\else\fi\

\begin{constructor} \label{const:Lex}
The \emph{lexicographic constructor} \ifshort receives \else receives as input \fi
implicit rankings $\rankname^i=(\conservedsuper{i},\reducedsuper{i})$ 
for $i=1,\ldots,m$, each with parameters $\seq x$.
It returns an implicit ranking $\mathrm{Lex}(\rankname^1,\ldots,\rankname^m) = ({\conserved},{\reduced})$ with parameters $\seq x$ defined by:
\ifshort$$\textstyle
    {\reduced}(\seq {x_\low},\seq {x_\high})= \bigvee_i
\left(\reducedsuper{i}(\seq {x_\low},\seq {x_\high}) \wedge \bigwedge_{j<i} 
 \conservedsuper{j}(\seq {x_\low},\seq {x_\high})
\right)\qquad
\conserved(\seq {x_\low},\seq {x_\high})  ={\reduced}(\seq {x_\low},\seq {x_\high}) \vee \bigwedge_i \conservedsuper{i}(\seq {x_\low},\seq {x_\high}) 
$$
\else
\begin{center}
   \begin{tabular}{rl}
    ${\reduced}(\seq {x_\low},\seq {x_\high})$&$= \bigvee_i
\left(\reducedsuper{i}(\seq {x_\low},\seq {x_\high}) \wedge \bigwedge_{j<i} 
 \conservedsuper{j}(\seq {x_\low},\seq {x_\high})
\right)$ \\
$\conserved(\seq {x_\low},\seq {x_\high})  $&$={\reduced}(\seq {x_\low},\seq {x_\high}) \vee \bigwedge_i \conservedsuper{i}(\seq {x_\low},\seq {x_\high}) $
\end{tabular}
\end{center}
\fi
\end{constructor}
\ifshort
\else
The soundness of the Lexicographic Constructor is obtained using the same ranking range and ranking function as for the pointwise constructor, except that the order on $A_1\times \cdots\times A_m$ is $\leq_{\text{lex}}$.
\fi

\subsection{Constructors for Domain-Based Aggregations of Rankings}\label{subsec:domain_based}
We now turn to present constructors that encode domain-based aggregations of rankings. 
As demonstrated in \Cref{example:toystab},
such aggregations can sometimes replace 
summations and cardinalities of sets in ranking arguments.
These are constructors that receive as input an implicit ranking $(\conserved,\reduced)$ with parameters $\seq x$, partitioned to $\seq x = \seq y \concatvar \seq z$. and return an implicit ranking with parameters $\seq z$, for which the conservation and reduction depend on the aggregation of the results of $\conserved$ and $\reduced$ for different values of $\seq y$ in the domain.
This is established by quantifying over $\seq y$.
To do so, we use generalizations of the partial-order constructions we saw in \Cref{subsec:finite_aggregations} to sets of functions. 
Due to the nature of these constructions, the soundness of the produced rankings depends on the finiteness of the domain, so these constructors are \emph{finite-domain constructors}.

\paragraph{\bf Domain-Pointwise Constructor.}
For a partially-ordered set $(A, \leq_A)$ %
and a set $Y$, 
we can lift $\leq_A$ to a pointwise partial order on
the set of functions $Y\to A$:
for any $a,b\in Y\to A$ we have
$a\leq_{\text{pw}} b \iff \forall y\in Y. a(y)\leq_A b(y)$.
If the order $\leq_A$ is a wfpo and $Y$ is finite then $\leq_{\text{pw}}$ is a wfpo.
In our case the set $Y$ is the set of assignments to $\seq y$ over some finite domain, and the next constructor encodes this idea in FOL in a straight-forward way: conservation of the aggregated ranking occurs when we have conservation of the input ranking for all values of $\seq y$, and reduction is achieved when, additionally, there is some value of $\seq y$ for which we have reduction of the input rank.
\ifshort A corresponding ranking function to $\assignset(\seq y,\domain)\to A$ is defined by $f(\struct,\assign)=\lambda \otherassign\in \assignset(\seq y, \domain). f_A(\struct,\otherassign\concatfunc\assign)$ where $f_A$ is a ranking function to $A$.\else\fi

\begin{constructor} 
\label{const:DomPW}
The \emph{domain-pointwise constructor} \ifshort receives \else receives as input \fi an
implicit ranking $\rankname^\insup=(\conservedsuper{\insup},\reducedsuper{\insup})$
with parameters $\seq x = \seq y\concatvar\seq z$.
It returns an implicit ranking for finite domains $\mathrm{DomPW}(\rankname^{\insup},\seq y) = ({\conserved},{\reduced})$ with parameters $\seq z$ defined by:
\ifshort
$$
{\conserved}(\seq {z_\low},\seq {z_\high}) = \forall {\seq y}. \conservedsuper{\insup}(\seq y \concatvar \seq {z_\low}, \seq y \concatvar \seq {z_\high}) \qquad
{\reduced}(\seq {z_\low},\seq {z_\high}) = {\conserved}(\seq {z_\low}, \seq {z_\high}) \wedge
(\exists \seq y. {\reducedsuper{\insup}}( \seq y \concatvar \seq {z_\low}, \seq y \concatvar \seq {z_\high} ))  
$$
\else
\begin{center}
\begin{tabular}{ll}
&${\conserved}(\seq {z_\low},\seq {z_\high}) = \forall {\seq y}. \conservedsuper{\insup}(\seq y \concatvar \seq {z_\low}, \seq y \concatvar \seq {z_\high})$\\
&${\reduced}(\seq {z_\low},\seq {z_\high}) = {\conserved}(\seq {z_\low}, \seq {z_\high}) \wedge
(\exists \seq y. {\reducedsuper{\insup}}( \seq y \concatvar \seq {z_\low}, \seq y \concatvar \seq {z_\high} ))  $
\end{tabular}    
\end{center}

\fi

\end{constructor}
\ifshort\else
To establish the soundness of $\rankname=\mathrm{DomPW}(\rankname^{\insup},\seq y)$, consider a domain $\domain$ and let $(A_\insup,\leq_\insup)$ and $f_\insup$ be a ranking range for $\rankname^\insup$. 
We define a ranking range for $\rankname$ by $(\assignset(\seq y,\domain)\to A_\insup,\leq_{\text{pw}})$, and a ranking function by $f(\struct,\assign)=\lambda \otherassign\in \assignset(\seq y, \domain). f_\insup(\struct,\otherassign\concatfunc\assign)$.
\fi
Next we demonstrate how DomPW can be used to approximate set cardinalities.%
\begin{example}\label{ex:DomPW}
    For a formula with one free variable $\formula(x)$, if we take $\rankname^{\insup}=\mathrm{Bin}(\alpha)$, the implicit ranking for finite domains $\rankname = \mathrm{DomPW}(\rankname^{\insup},x)=(\conserved,\reduced)$
    \ifshort approximates \else can be seen to 
    approximate \fi the cardinality of the set of elements that satisfy $\formula$:
    if \ifshort \else two \fi structures $\struct_\low,\struct_\high$ are such that $\struct_\low,\struct_\high\models \reduced$, we have that the set of elements that satisfy $\alpha$ in $\struct_\low$ is a strict subset of that in $\struct_\high$, and so we have reduction in cardinality.
    This does not capture cardinality precisely, and we can encode a more precise approximation using \Cref{const:DomPerm}.

\end{example}
\begin{example}
    \label{ex:SetFreeVar}
    Taking $\rankname^{\insup} = \mathrm{Bin}(\mathrm{priv}(i)\wedge \mathrm{lt}(i,j))$ from \Cref{ex:Bin}, the implicit ranking for finite domains $\mathrm{DomPW}(\rankname^\insup,j)=(\conserved,\reduced)$ aggregates over $j$ but still has $i$ as a parameter. This lets us compare cardinalities of the sets of machines associated with two different machines described by $i_\low$ and $i_\high$.
\end{example}

\paragraph{\bf Domain Permuted-Pointwise Constructor.}
The following constructor is a relaxation of \Cref{const:DomPW}, based on the notion of a permuted pointwise order, meant to  
capture cases where two functions are almost pointwise-ordered but some permutation of the inputs is required.
For a set $Y$, a partially-ordered set $A$, and two functions $a,b\in Y\to A$, we say that $a \leq_{\text{perm}} b$ if there exists a bijection $\bijec$ on $Y$ such that $a \leq_{\text{pw}} b\circ \bijec$. 
The result,  
$\leq_{\text{perm}}$, is a \emph{preorder} on $Y\to A$. It induces a partial order on the quotient set of $Y\to A$ w.r.t.\  $\equiv_{\text{perm}}$ in the usual way. If $(A,\leq_A)$ is a wfpo and $Y$ is finite, $\leq_{\text{perm}}$ is a wfpo on the quotient set of $Y\to A$ (see \Cref{appendix:proofs} for more details).
\ifshort In our case $Y = \assignset(\seq y, \domain)$ and a ranking function $f_A$ to $A$ can be lifted to a ranking function to the quotient set of $\assignset(\seq y, \domain) \to A$ by defining $f(\struct,\assign)= [\lambda \otherassign\in \assignset(\seq y, \domain). f_A(\struct,\otherassign\concatfunc\assign)]_{\equiv_{\text{perm}}}$. \else\fi

The above order is not directly first-order definable --- we cannot capture the existence of a permutation $\bijec$ as this would require second order quantification.
Instead, the permuted-pointwise constructor under-approximates and encodes only the cases where $\bijec$ is 
composed of transpositions of at most a constant number $k$ of pairs of elements $\seq {y^1_{\scriptstyle\rightarrow}},\seq {y^1_{\scriptstyle\leftarrow}},\ldots,\seq {y^k_{{\scriptstyle\rightarrow }}},\seq {y^k_{\scriptstyle\leftarrow}}$ as follows: $\bijec(\seq {y^i_{\scriptstyle\rightarrow}})=\seq {y^i_{\scriptstyle\leftarrow}}$
and $\bijec(\seq {y^i_{\scriptstyle\leftarrow}})=\seq {y^i_{\scriptstyle\rightarrow}}$ for $i=1,\ldots,k$ and $\bijec(\seq y)=\seq y$ for any other $\seq y$.
To ensure that $\bijec$ is a well-defined permutation we require that for every $i\neq j$ we have ${\seq y^i_{\scriptstyle\rightarrow}}\neq {\seq y^j_{\scriptstyle\rightarrow}}$,${\seq y^i_{\scriptstyle\leftarrow}}\neq {\seq y^j_{\scriptstyle\leftarrow}}
     $ and $
     {\seq y^i_{\scriptstyle\rightarrow}}\neq {\seq y^j_{\scriptstyle\leftarrow}}$.
We can then capture $\bijec(\seq y)$ for any $\seq y$ by a term $\seq y_\bijec$ (see below), and encode reduction or conservation according to $\bijec$ by comparing $\seq y$ to $\seq y_\bijec$ with the input implicit ranking. (Notably, this is the only constructor that uses the input implicit ranking to compare different elements.) 
While this is only an approximation,
we have found it captures several interesting cases, such as the one needed to verify \Cref{example:toystab}.
\begin{constructor} \label{const:DomPerm}
The \emph{domain permuted-pointwise constructor} \ifshort receives \else receives as input \fi an
implicit ranking $\rankname^\insup=(\conservedsuper{\insup},\reducedsuper{\insup})$
with parameters $\seq x = \seq y\concatvar\seq z$, and $k\in\nat$.
It returns an implicit ranking for finite domains $\mathrm{DomPerm}(\rankname^\insup,\seq y,k)=({\conserved},{\reduced})$ with parameters $\seq z$ defined by:
\begin{center}
\begin{tabular}{l}
     $\conserved(\seq {z_\low},\seq {z_\high}) = \tilde \exists\bijec. \
    \forall \seq {y}. \ \conservedsuper{\insup}(\seq y \concatvar \seq {z_\low}, \seq {y_\bijec} \concatvar \seq {z_\high} )$\\
    $ \reduced(\seq {z_\low},\seq {z_\high}) = \tilde \exists\bijec. \ \forall \seq {y}. \ \conservedsuper{\insup}(\seq y \concatvar \seq {z_\low}, \seq {y_\bijec} \concatvar \seq {z_\high} ) \wedge 
    \exists \seq {y}. \ \reducedsuper{\insup}(\seq y \concatvar \seq {z_\low}, \seq {y_\bijec} \concatvar \seq {z_\high} )$
\end{tabular}
\end{center}
where:\\[8pt]
\begin{tabular}{l}
     $\mathrm{ \tilde \exists\bijec. \ \formula } := \exists \seq {y^1_{\scriptstyle\rightarrow}},\seq {y^1_{\scriptstyle\leftarrow}},\ldots,\seq {y^k_{\scriptstyle\rightarrow}},\seq {y^k_{\scriptstyle\leftarrow}}. \bigwedge_{1\leq i < j\leq k} ({\seq y^i_{\scriptstyle\rightarrow}}\neq {\seq y^j_{\scriptstyle\rightarrow}}\wedge {\seq y^i_{\scriptstyle\leftarrow}}\neq {\seq y^j_{\scriptstyle\leftarrow}}
     \wedge
     {\seq y^i_{\scriptstyle\rightarrow}}\neq {\seq y^j_{\scriptstyle\leftarrow}}

     )\wedge \formula$ \\
     $\seq y_\bijec =  \mathrm{ite}(\seq y = \seq {y^1_{\scriptstyle\rightarrow}},\
\seq {y^1_{\scriptstyle\leftarrow}}, \
\mathrm{ite}(\seq y = \seq {y^1_{\scriptstyle\leftarrow}},\
\seq {y^1_{\scriptstyle\rightarrow}} , \ 
\ldots, \ 
\mathrm{ite}(\seq y = \seq {y^k_{\scriptstyle\rightarrow}}, \ \seq {y^k_{\scriptstyle\leftarrow}}, \ \mathrm{ite}(\seq y = \seq {y^k_{\scriptstyle\leftarrow}}, \ \seq {y^k_{\scriptstyle\rightarrow}} , \ \seq y))))$
\end{tabular}
\end{constructor}
\ifshort\else
To establish the soundness of the domain permuted-pointwise constructor, we take the same construction of ranking function defined for \Cref{const:DomPerm}, modulo an equivalence relation of functions up to permutation.
\fi
Next, we demonstrate how we can use DomPerm to approximate sums over unbounded sets and weighted set cardinalities. 
We expand on the relation to summation in \Cref{appendix:heights}.

\begin{example}
    \label{ex:DomPerm}
    Continuing with \Cref{ex:SetFreeVar}, taking $\rankname^{\insup} = \mathrm{DomPW}(\mathrm{Bin}(\mathrm{priv}(i)\wedge \mathrm{lt}(i,j)),j)$,
    $\rankname = \mathrm{DomPerm}(\rankname^{\insup},i,1)$  
    is an implicit ranking that captures the reduction argument  described in \Cref{example:toystab}. 
    In particular, the aggregation captured by the permuted-pointwise constructor replaces an unbounded summation. 
    \ifshort 
The reduction formula listed in \Cref{example:toystab}  is slightly simplified
compared to $\rankname$ produced by  $\mathrm{DomPerm}$ above: it includes only %
the nontrivial disjuncts of the $\mathrm{ite}$ expressions, and it  uses $\mathrm{skd},\mathrm{skd}+1$ in place of the existential quantifier. 

    \else Note that, compared to $\rankname$ above, the reduction formula presented in \Cref{example:toystab}  is  simplified in two ways: 
    we consider only some of the disjuncts produced by the $\mathrm{ite}$ expressions, and we substitute $\mathrm{skd},\mathrm{skd}+1$ for the existential quantifiers produced by the $\mathrm{DomPerm}$ constructor. Such instantiations of existenial quantifiers turn out to be useful in practice, as we discuss in \Cref{sec:evaluation}.
    \fi
\end{example}

\begin{example}
\label{ex:Weighted}
Consider a ranking function $|\{x\mid \alpha(x)\}|+2|\{x\mid \beta(x)\}|$, where $\alpha$ and $\beta$ are two predicates, and consider the case where conservation of ranking between states $\struct_\low,\struct_\high$ holds due to `mixing' of the predicates, for example ``exchanging'' one element $x^0$ that satisfies $\beta$ in $\struct_\low$  with two elements $x^1,x^2$ that satisfy $\alpha$ in $\struct_\high$.
We can use the DomPerm constructor to capture such a ranking argument.
We first add to the signature of the transition system an enumerated sort \textbf{type} with three values: $\mathrm{type}_\alpha,\mathrm{type}_\beta^1,\mathrm{type}_\beta^2$. We then define the unified predicate $\gamma(x,ty)=(ty=\mathrm{type}_\alpha\wedge\alpha(x))\vee ( (ty=\mathrm{type}_\beta^1 \vee  ty=\mathrm{type}_\beta^2 )\wedge \beta(x) )$. This can be understood as  introducing one copy of the domain for $\alpha$ and two copies for $\beta$,  
and interpreting each of $\alpha$ and $\beta$ over the copies relevant to them. Finally, we consider $\rankname = \mathrm{DomPerm}(\mathrm{Bin}(\gamma),(x,ty),2)=(\conserved
,\reduced)$. 
To see that the aforementioned pair of states indeed satisfies the resulting conservation formula, i.e., $(\struct_\low,\struct_\high)\models \conserved$, consider the permutation defined by $\seq {y^1_{\scriptstyle\leftarrow}}=(x^0,\mathrm{type}^1_\beta), 
\seq {y^2_{\scriptstyle\leftarrow}}=(x^0,\mathrm{type}^2_\beta),
\seq {y^1_{\scriptstyle\rightarrow}}=(x^1,\mathrm{type}_\alpha), \seq {y^2_{\scriptstyle\rightarrow}}=(x^2,\mathrm{type}_\alpha)$.

\end{example}

\begin{remark}
    Our decision to focus on permutations obtained by transpositions in our first-order encoding was motivated by examples such as \Cref{ex:Weighted}. There are of course other classes of permutations that can be encoded in first-order logic and would also result in a sound approximation of $\leq_{\text{perm}}$.
\end{remark}

Note that \Cref{const:DomPW} (domain-pointwise) can be understood as a special case of \Cref{const:DomPerm} obtained by considering the degenerate case $k=0$.

\paragraph{\bf Domain-Lexicographic Constructor.}
The following constructor is an analog of \Cref{const:Lex}, where instead of aggregation based on the order of indices of the given rankings, we aggregate based on a partial order on $\assignset(\seq y,\domain)$. 
To that end we rely on an already-existing order in the system, encoded  
by a single signature formula $\orderformula(\seq {y_\low},\seq {y_\high})$, and axiomatized 
as in \Cref{const:Pos}.

For a partially-ordered set $A$ with 
a partial order $\leq_A$ and a set $Y$ with a wfpo $\leq_Y$,
the set of functions $Y\to A$ can be ordered by the lexicographic partial order:
for any $a,b\in Y\to A$ we have 
$a\leq_{\text{lex}}~b \iff a(y) \leq_A b(y)$ for all minimal elements of the set $\{y\in Y \mid a(y)\neq b(y)\}$ (for $a\neq b$, this set is not empty, because $\leq_{Y}$ is a wfpo)\ifshort\else\footnote{
This definition is more complicated than the standard definition since we allow $\leq_Y$ to be any wfpo, and not necessarily a well-order.}\fi.
Equivalently, for every $y\in Y$ such that $a(y)\not\leq_A b(y)$ there exists $y^*$ such that $y^* <_Y y$ with $a(y^*) <_A b(y^*)$.
Additionally, if $\leq_A$ is a wfpo and $Y$ is finite,
$\leq_{\text{lex}}$ is a wfpo.
The constructor encodes this in a straight-forward way: the set $Y$ is $\assignset(\seq y, \domain)$ and $<_Y$ is given by $\orderformula$.
\ifshort A corresponding ranking function can be defined by the interpretation of $\orderformula$ as in \Cref{const:Pos} combined with $\lambda \otherassign\in \assignset(\seq y, \domain). f_A(\struct,\otherassign\concatfunc\assign)$ as in \Cref{const:DomPW}.\else\fi

\begin{constructor} \label{const:DomLex}
The \emph{domain-lexicographic constructor} \ifshort receives \else receives as input \fi an
implicit ranking $\rankname^\insup=(\conservedsuper{\insup},\reducedsuper{\insup})$
with parameters $\seq x = \seq y\concatvar\seq z$, and a formula $\orderformula(\seq {y_\low},\seq {y_\high})$ over $\Sigma$.
It returns an implicit ranking for finite domains $\mathrm{DomLex}(\rankname^\insup,\seq y, \orderformula)=({\conserved},{\reduced})$ with parameters $\seq z$ defined by:
\begin{align*}
    {\conserved}(\seq {z_\low},\seq {z_\high}) & = \immutorder(\orderformula)\wedge \forall \seq{y}.  (\conservedsuper{\insup}(\seq {y}\concatvar \seq {z_\low},
    \seq {y}\concatvar \seq {z_\high})
    \vee \exists \seq {y^*}. (
    \orderformula_\low(\seq{y^*},\seq{y})
    \wedge
    \reducedsuper{\insup}(\seq {y^*}\concatvar \seq {z_\low},
    \seq {y^*}\concatvar \seq {z_\high})
    ) 
    )
     \\
    {\reduced}(\seq {z_\low},\seq {z_\high}) & = {\conserved}(\seq {z_\low}, \seq {z_\high}) \wedge
(\exists \seq y. {\reducedsuper{\insup}}( \seq y \concatvar \seq {z_\low}, \seq y \concatvar \seq {z_\high} )) 
\end{align*}
\end{constructor}
\ifshort \else For the soundness proof of the domain lexicographic constructor we can use a  construction of ranking range and ranking function similar to the one of \Cref{const:DomPW}, combined with the ideas of the partial order presented in the soundness proof of \Cref{const:Pos}.
\fi

\begin{example}\label{ex:DomLex} Continuing with \Cref{example:binary_counter}, define a formula $\formula(i)$ which is set to true if the array holds a $1$ in index $i$, and take $\rankname^{\insup}  = \mathrm{Bin}(\formula)$. 
Take, as in \Cref{ex:Pos}, $\orderformula(i_\low,i_\high) = \mathrm{lt}(i_\low,i_\high)$. Then, $\rankname = \mathrm{DomLex}(\rankname^{\insup},i,\orderformula)$ captures lexicographic reduction of values in the array, akin to reduction in a binary counter.
Compared to the formula given above, the reduction formula presented in \Cref{example:binary_counter} is simplified by the additional assumption that the order given by $\ell$ is total, which need not be the case in general.
\end{example}
Note that \Cref{const:DomPW} (domain-pointwise) can be understood as a special case of \Cref{const:DomLex} obtained by considering the degenerate case $\orderformula(\seq {y_\low},\seq {y_\high}) = \mathrm{false}$.

%% file: sections/evaluation.tex
\section{Implementation and Evaluation}\label{sec:evaluation}

To explore the power of the implicit rankings defined by our constructors we implemented
a deductive verification procedure 
for liveness properties in python, using the Z3 API~\cite{z3}.
The procedure is based on a proof rule, given in \Cref{appendix:proof_rule}, that uses implicit rankings to prove liveness, inspired by existing proof rules.
Our implementation is available at \cite{github}.
It takes as input a first-order transition system, a liveness property of the form $(\bigwedge_i \forall \seq x \globally \eventually r_i(\vec x)) \to  \globally (p \to\eventually q)$ where $r_i$ are parameterized  fairness assumptions, a closed implicit ranking  defined using the constructors of \Cref{sec:constructions} 
and the other formulas required for applying the proof rule. 
Given the above,  we automatically construct the implicit ranking formulas $\reduced,\conserved$ defined by the constructors and use Z3 to validate the premises of the rule. 
Some of the domain-based constructors create formulas with quantifier alternations, which may be challenging for solvers.   
We thus allow the user to provide \emph{hint terms} for the existential quantifiers in the declaration of such constructors. 
We then replace the existential quantification with disjunction over formulas substituted with these terms in the solver queries (this is sound since the implicit ranking formulas appear only positively in the proof rule and in recursive constructors).

\subsubsection*{Results.}
We evaluate our tool on a suite of examples from previous works, listed below. 
We use Z3 version 4.12.2.0, run on a laptop running Windows with a Core-i7 2.8 GHz CPU. All examples are successfully verified within 10 minutes.
In some examples, such as our two motivating examples, there are no fairness assumptions, in which case we use
$\globally \eventually \mathrm{true}$ as a fairness assumption. 
This amounts to assuming totality of the transition relation, which indeed holds in these examples.
We expand on the issue of totality in \Cref{appendix:totality}.
Some of the more complicated examples required some abstraction techniques, which we describe in \Cref{appendix:abstractions}. 
Next we describe the examples. For each example, we present the 
composition of the constructors that define the implicit ranking. We omit the other arguments of the constructors. Additionally, we note whenever the validation of the premises required user-provided hint terms for the existential quantifiers introduced by the constructors. 
In all examples we use finite-domain constructors, which assume a finite (but unbounded) semantics for the domain of the relevant sorts.

\textit{Toy Stabilization.} (\Cref{example:toystab}) The ranking argument described in \Cref{sec:motivating_examples} is captured by an implicit ranking  defined by DomPerm(DomPW(Bin)).

\textit{Binary Counter.} (\Cref{example:binary_counter})
The ranking argument described in \Cref{sec:motivating_examples} is captured by an implicit ranking defined using Lex(DomLex(Bin),Pos).

\textit{Mutex Ring.} Taken from \cite{liveness_invisible}, a simple mutual exclusion protocol, where a token moves around a ring, allowing any machine that holds it to enter the critical section. The liveness property is that every machine that tries to enter the critical section eventually does.
We use an implicit ranking %
defined by DomLex(Lex(Bin,Bin,Bin)), where Lex(Bin,Bin,Bin) is used to capture the local state of a machine and the domain-lexicographic aggregation captures the linear movement of the token between machines along the 
ring.
In \Cref{appendix:linsums} we discuss a different implicit ranking, based on the DomLin constructor.

\textit{Leader Ring.} The Chang-Roberts algorithm for leader election in a ring~\cite{chang_roberts}, with first-order modeling and invariants based on~\cite{bounded_horizon,ivy}.
The liveness property is that eventually a leader is elected. We use an implicit ranking defined by
Lex(DomPW(Bin),DomLex(Bin)),
where DomPW(Bin) is used to track the cardinality of the set of machines that did not yet send their id, and 
DomLex %
tracks the movement of messages around the ring by aggregating a binary ranking based on the formula $\formula(i,j)=\mathrm{pending}(\mathrm{id}(i),j)$ lexicographically according to an order on $(i,j)$ that depends only on the ring order of $j$.

\textit{Self-Stabilization Protocols.}
We verify several self-stabilizing protocols~\cite{dijkstra_self_stab}. In these protocols, a set of machines is organized in a line/ring. 
Each machine holds some local value(s). Privileges are assigned to machines according to 
derived relations between a machine's values and its neighbors' values. In every transition a privileged machine changes its local values, with the effect that it loses its privilege and a new privilege may be created for one of its neighbors.
The protocols differ in the privileges they use and in the way the local states of machines are defined and updated.
The desired properties of a self-stabilizing protocol are:
(i) eventually a unique privilege is present (stabilization), 
(ii) starting with a unique privilege, uniqueness is maintained globally (maintenance), and 
(iii) starting with a unique privilege, every machine gets a privilege infinitely often (fairness).
We verify properties (i) and (iii) which are liveness properties.%

\textit{Dijkstra's $k$-State Protocol}. In this protocol every machine holds a value taken from a set whose cardinality is larger than the number of machines. 
A designated machine $\mathrm{bot}$ %
can introduce new values into the ring.
We prove stabilization by proving three lemmas (based on \cite{regular_abstractions},\cite{mechanized_k_state_stab}): 
(a) machine $\mathrm{bot}$ is eventually scheduled, 
(b) if $\mathrm{bot}$ is scheduled infinitely often then eventually $\mathrm{bot}$ holds a unique value in the ring and 
(c) if $\mathrm{bot}$ holds a unique value in the ring eventually there is a unique privilege in the ring.
The implicit ranking required for properties (a) and (c) is the same and has structure DomLex(Bin), tracking the movement of privileges towards $\mathrm{bot}$ by aggregation over the machines according to the ring order.
(We could have also used the implicit ranking of \Cref{example:toystab}.)
The implicit ranking for property (b) has structure DomPW(Bin), aggregating over values to
capture the distance between $\mathrm{bot}$'s value and a new value.
For fairness, a similar ranking is used as for properties (a) and (c), tracking the movement of a privilege towards the machine for which fairness is shown.

\textit{Dijkstra's $4$-State Protocol.} In this protocol every machine holds two binary values %
and two kinds of privileges can be derived: from above and from below.
The different privileges move between machines in different directions. We base our proof on observations in \cite{dijkstra_4state_stab_note}.
The implicit ranking we use for proving stabilization is defined by a lexicographic pair $\mathrm{Lex}(\rankname^1,\rankname^2)$ where $\rankname^1$ has structure DomPerm(Bin), capturing the number of privileges of both kinds. $\rankname^2$ has structure DomPerm(PW(DomPW(Bin),DomPW(Bin)))
which intuitively gives a bound on the number of moves required from all machines until a privilege is lost. To that end,
DomPerm is used as in \Cref{example:toystab} to approximate the sum of moves required from a single machine.
This number can be at most twice the number of machines.
Therefore we use
two DomPW(Bin) rankings, each capturing a cardinality of a set of machines. The two are composed by a PW constructor.
For this implicit ranking we used simple hint terms as described above, replacing existential quantification with instantiation of $\mathrm{skd},\mathrm{skd.next}$ or $\mathrm{skd.prev}$.
The structure of implicit ranking we use for fairness is somewhat simpler: $\mathrm{Lex}(\rankname^1,\rankname^2,\rankname^3)$ with each $\rankname^i$ tracking a movement of some privilege in some direction using $\mathrm{DomLex(Bin)}$.

\textit{Ghosh's $4$-State Protocol.} A simplification of Dijkstra's $4$-state protocol, where every machine holds a value in $\{0,1,2,3\}$, and two kinds of privileges are derived.
The implicit ranking we use to prove both stabilization and fairness is similar to the respective ranking for Dijkstra's $4$-state protocol, except that we have a lexicographic component that encodes a reduction in the number of \emph{breaks} which are successive nodes that hold different values (see~\cite{ghosh_stab}).

\textit{Dijkstra's $3$-State Protocol.} 
In this protocol every machine holds a value in $\{0,1,2\}$ and two kinds of privileges are derived.
The implicit ranking we use for proving stabilization is a direct encoding of the ranking function given in \cite{kessels_3state_stab_proof}, which itself is based on observations in \cite{dijkstra_3state_stab_proof}. 
This ranking is a 4-argument lexicographic ranking $\mathrm{Lex}(a,b,c,d)$. 
Ranking $a$ captures a weighted sum of set cardinalities using an implicit ranking with structure $\mathrm{DomPerm}(\mathrm{Bin})$ with $k=4$ as shown in \Cref{ex:Weighted}, notably this generates premises with 18 existential quantifiers which require relatively complicated hints to validate.
Rankings $b$ and $c$ capture set cardinalities encoded using $\mathrm{DomPerm}(\mathrm{Bin})$ and ranking $d$ captures the number of moves from any machine until reduction in either $a,b,c$, similarly to the $4$-state protocol, encoded using composition of DomPerm and DomPW.%
The implicit ranking we use for fairness is similar to Dijkstra's $4$-state protocol.

%% file: sections/related.tex
\section{Related Work and Concluding Remarks}\label{sec:related}

Numerous approaches for liveness verification have been proposed in the literature, including
abstraction techniques~\cite{regular_abstractions,abstraction_refinement},
liveness to safety reductions~\cite{liveness_to_safety,liveness_to_safety_orig,liveness_to_safety_via_implicit_abstraction,prophecy}, rich proof structures~\cite{liveness_parametrized_programs,fairness_modulo_theory,automata_approach}, and more~\cite{ironfleet,anvil}.
Our work considers verification based on the classical notion of ranking functions~\cite{checking_large,meanings_to_programs}. 
Many proof rules based on ranking functions have been suggested, %
e,g.,~\cite{completing_temporal,towards_liveness_proofs,liveness_invisible,transition_invariants,termination_systems_code}.
As explained in \Cref{subsec:using_ranking}, our constructions of implicit rankings can be used in any rule that requires conservation and reduction of rank.
In contrast to~\cite{lvr,synthesis_linear_ranking,ordinal_valued_ranking,synthesis_by_bits,linear_ranking_lasso,proving_lagrangian} which automate the search for ranking functions, we focus on expanding the class of ranking functions that can be used and leave automation for future work.
We now turn to expand on the most relevant recent works.

The closest work to ours is~\cite{towards_liveness_proofs}, which 
uses \emph{relational rankings} or fixed-size lexicographic tuples of  them as ranking functions in liveness proofs for first-order transition systems. Relational rankings count the number of elements that satisfy some predicate, and their reduction and conservation are measured approximately in a pointwise fashion. 
Relational rankings can be captured by our constructors, specifically $\mathrm{DomPW}(\mathrm{Bin})$, but our constructors induce a richer family of rankings and offer more expressiveness, sometimes, at the expense of more complex quantifier structure. 
On the other hand,~\cite{towards_liveness_proofs} does not assume finiteness, but proves that the set considered is finite at any time. Generalizing our finite-domain constructors to this setting is not trivial, and is a subject for future work.
The proof rule of~\cite{towards_liveness_proofs} also allows modular temporal reasoning.  The definition of implicit ranking we offer can also be used in such a proof rule.

The approach of \cite{lvr} focuses on automating the search for a ranking function for  protocols modeled in FOL. They automatically synthesize integer-valued polynomial ranking functions from integer variables that appear in the protocol specification, fairness variables \ifshort \else that they manually incorporate to encode fairness\fi, scheduling variables, and variables that reflect cardinalities of sets defined by predicates (some of which are user-provided).
Similarly to our domain-based aggregations, they assume finiteness of certain domains.
While we use the assumption implicitly, they introduce integer variables that bound the finite cardinalities.
In both~\cite{towards_liveness_proofs} and~\cite{lvr}, only rankings that are polynomial in the cardinality of the domain can be used, which cannot capture \Cref{example:binary_counter}.
Additionally, their domain-based aggregations, which are limited to cardinalities of sets, cannot be recursively composed, which is needed for examples such as \Cref{example:toystab}.%

The approach of~\cite{liveness_to_safety,prophecy} is based on a liveness-to-safety reduction which uses a dynamic finite abstraction to reduce liveness to acyclicity of traces. The reduction is encoded via a monitor that augments the original transition system. 
Liveness is established by showing that 
an arbitrary monitored state called the saved state is never revisited.
While extremely powerful, the approach is difficult to use since one has to find an invariant of the augmented transition system that justifies that the current state is never equal to the saved state.
Our implicit rankings can be understood as a natural way to describe such two-state invariants without having to reason about an augmented transition system. This also makes them more amenable to automation. 

\subsubsection{Acknowledgement}
We thank Neta Elad and Eden Frenkel for helpful discussions.
The research leading to these results has received funding from the
European Research Council under the European Union's Horizon 2020 research and innovation programme (grant agreement No [759102-SVIS]).
This research was partially supported by the Israeli Science Foundation (ISF) grant No.\ 2117/23.

%% file: appendices/alternative_defs.tex
\section{Alternative Definitions for Implicit Ranking}\label{appendix:definition}

In \Cref{def:implicit_ranking} we define formulas both for conservation and for reduction. In the case of a partial order $\leq$, we can derive the strict order $<$ from $\leq$ by considering $\leq \cap \not \geq$. This might suggest that we could use $\conserved(\seq {x_\low},\seq {x_\high}) \wedge \neg \conserved(\seq {x_\high},\seq {x_\low})$ for $\reduced(\seq {x_\low},\seq {x_\high})$. However, since $\conserved$ only underapproximates $\leq$, it
is sometimes the case that the above does not guarantee reduction.
This is the case, for example, in the implicit ranking used to underapproximate reduction/conservation of an unbounded sum in \Cref{example:toystab} (see \Cref{const:DomPerm}).

A natural alternative to \Cref{def:implicit_ranking} is a definition that requires that $\conserved,\reduced$ define a wfpo on the set of {\pairnames} itself, instead of mapping to a set $A$ with its own wfpo.
This is a strictly stronger definition: if $\conserved,\reduced$ define a wfpo on the set of
{\pairnames}, then we can take the ranking function as the identity.
On the other hand, a natural way to derive relations on the set of {\pairnames} is by a \textit{pullback}. That is: $R_\leq = \{ (\struct_\low,\assign_\low),(\struct_\high,\assign_\high) \mid f(\struct_\low,\assign_\low) \leq f(\struct_\high,\assign_\high) \}$.
It is easy to verify that $R_\leq$ is always transitive, but unfortunately it need not be weakly antisymmetric. If we define $R_<$ similarly we do get a strict well-founded order. If we then attempt to define $R_\leq = R_< \cup \text{Id}$, conservation would not necessarily hold.
Recall that because we are interested in liveness under fairness and approximations of orders we must define both $\leq$ and $<$.

This gives rise to the ``correct'' alternative definition: $\conserved,\reduced$ are an implicit ranking if there exist relations on pairs of {\pairnames} $R_\leq, R_<$ such that the set of models of $\conserved$ and $\reduced$ are contained in $R_\leq$ and $R_<$ respectively and additionally: \begin{itemize}
    \item $R_\leq$ is a pre-order (transitive and reflexive).
    \item $R_<$ is antisymmetric, transitive, and well-founded.
    \item $R_\leq \circ R_< \subseteq R_<$ and $R_< \circ R_\leq \subseteq R_<$. 
\end{itemize}
Importantly, as before, we only require the set of models to be contained in $R_\leq$ and $R_<$ and not necessarily equal to allow approximations. 
Another note is that transitivity (and 
similarly, reflexivity of $R_\leq$) does not need to be checked as $R_\leq$ and $R_<$ can be replaced by their transitive closure without loss of the other properties. 
We give \Cref{def:implicit_ranking} and not the alternative as it is more similar to the classic notions of ranking functions, and easier to show in the proofs, without needing to go through the pullback. 
This alternative definition fits with definitions such as that of \cite{transition_invariants}.

There is a connection between transitivity of the induced relations from $\conserved,\reduced$ and whether or not they are approximate or capture an exact conservation or reduction.
If $\conserved,\reduced$ are transitive, one can define the ranking range to be the quotient set of the set of {\pairnames} by the equivalence relation derived from the pre-order $\conserved$ with $f$ being the natural mapping.
Then $\conserved,\reduced$ exactly capture conservation and reduction of $f$. 
In the other direction, if $\conserved,\reduced$ exactly capture reduction of some ranking function $f$ then it follows they must be transitive from transitivity of the wfpo on the ranking range.
One can verify that formulas $\conserved,\reduced$ of implicit rankings constructed using \Cref{const:DomPerm} are not transitive --- they are strict under-approximations. 
Our other constructors are not approximate.

%% file: appendices/linear_sum_constructors.tex
\section{Linear Sum Constructors}\label{appendix:linsums}

In this section we give two more constructors, which are based on the notion of the linear sum of partial orders.
The first aggregates finitely-many rankings, similarly to \Cref{const:Lex,const:PW} and the second is a domain-based analog of the first, similarly to \Cref{const:DomLex,const:DomPW,const:DomLex}.
The soundness proofs are given with the soundness proofs of the other constructors in \Cref{appendix:proofs}.
While these constructors are natural, they add less expressiveness than the ones we saw before, and can sometimes be replaced by the corresponding lexicographic constructors. 
We managed to verify all considered examples without these constructors, but we have implemented them as well.

\paragraph{\bf Linear Sum Constructor.}
For partially-ordered sets $A_1,\ldots,A_m$ with partial orders $\leq_1,\ldots,\leq_m$ respectively, one can define the linear sum (or ordinal sum) partial order $\leq_\oplus$ on the set $A_1\uplus\cdots \uplus A_m$ by
$ a \leq_{\oplus} b \iff \bigvee_i (a,b \in A_i \wedge a\leq_i b) \vee \bigvee_{i,j} (a\in A_i, b \in A_j \wedge i < j)$.
Intuitively, elements in the same set are ordered by the order of that set, and elements in different sets are ordered by the index of their set.
As before, if the orders $\leq_i$ are all wfpos then so is $\leq_\oplus$.

The linear sum constructor uses this definition for aggregation by partitioning the {\pairnames} into disjoint sets and ranking each part using a different given ranking $\rankname^i$.
The partition is defined by formulas $\formula^i$,  
with the intention that the part an {\pairname} belongs to is determined by the minimal formula $\formula^i$ it satisfies. Minimality is encoded by $\formulaa^i = \formula^i \wedge \bigwedge_{j < i} \neg \formula^j$.

\begin{constructor}  \label{const:Lin}
The \emph{Linear Sum Constructor} \ifshort receives \else receives as input \fi
implicit rankings $\rankname^i=(\conservedsuper{i},\reducedsuper{i})$ 
for $i=1,\ldots,m$, each with parameters $\seq x$,
and formulas $\formula^i(\seq x)$ over $\Sigma$.
It returns an implicit ranking $\mathrm{Lin}(\rankname^1,\ldots,\rankname^m,\formula^1,\ldots,\formula^m) = ({\conserved},{\reduced})$ with parameters $\seq x$ defined by:
    \begin{center}
        \begin{tabular}{rl}
&$\formulaa^i(\seq x) = \formula^i (\seq x) \wedge \bigwedge_{j<i} \neg \formula^j (\seq x)$\\
&${\conserved}(\seq {x_\low},\seq {x_\high}) = \left(\bigvee_{i} \conservedsuper{i}(\seq {x_\low},\seq {x_\high}) \wedge \formulaa_\low^i(\seq {x_\low}) \wedge \formulaa_\high^i(\seq {x_\high})\right) \vee \left(\bigvee_{i<j} \formulaa_\low^i(\seq {x_\low}) \wedge \formulaa^j_\high(\seq {x_\high}) \right)$\\
&$\reduced(\seq {x_\low},\seq {x_\high}) = \left(\bigvee_{i} \reducedsuper{i}(\seq {x_\low},\seq {x_\high}) \wedge \formulaa_\low^i(\seq {x_\low}) \wedge \formulaa_\high^i(\seq {x_\high})\right) \vee \left(\bigvee_{i<j} \formulaa_\low^i(\seq {x_\low}) \wedge \formulaa_\high^j(\seq {x_\high}) \right)$
\end{tabular} 
\end{center}

\end{constructor}

\paragraph{\bf Domain-Linear Sum Constructor.}
The generalization of the linear sum partial-order to the disjoint union of an unbounded number of sets gives a lexicgoraphic order on the set $Y\times A$. The first component determines a set in the disjoint union and the second compares within a set. If both $\leq_A$ and $Y$ are wfpo, 
$\leq_{\text{lex}}$ is a wfpo on $Y\times A$.
The first-order logic encoding is similar to the one used in \Cref{const:Lin}:
we partition the set of {\pairnames} into disjoint sets, and compare {\pairnames} within the same part using a given implicit ranking. 
However, here, the implicit rankings used for each part are the \emph{same} implicit ranking, instantiated with different inputs for $\seq y$. 
Further, the partition is defined using a \emph{single} formula $\formula(\seq y)$, with the intention that the part an {\pairname} belongs to is determined by the minimal value of $\seq y$ for which it satisfies $\formula$. Minimality is measured according to the order on values of $\seq y$ induced by $\orderformula$. This condition is encoded by the formula $\formulaa$.

\begin{constructor} \label{const:DomLin}
The \emph{Domain Linear Sum} \ifshort receives \else receives as input \fi an
implicit ranking $\rankname^\insup=(\conservedsuper{\insup},\reducedsuper{\insup})$
with parameters $\seq x=\seq y\concatvar\seq z$, a formula $\orderformula(\seq {y_\low},\seq {y_\high})$ over $\Sigma$, and a formula $\formula(\seq x)$ over $\Sigma$.
It returns an implicit ranking for finite domains $\mathrm{DomLin}(\rankname^\insup,\seq y, \orderformula,\formula)=({\conserved},{\reduced})$ with parameters $\seq z$ defined by:
\begin{align*}
& \formulaa(\seq y \concatvar \seq z) = \formula(\seq y\concatvar\seq z) \wedge \forall {\seq {y''}}.( \orderformula(\seq {y},\seq {y''})\vee (\seq y = \seq {y''}) \vee \neg \formula(\seq {y''}\concatvar \seq {z}))\\
&\reduced(\seq {z_\low},\seq {z_\high}) =  \ \immutorder(\orderformula) \wedge (\exists \seq {y}. (\reducedsuper{\insup}(\seq y \concatvar \seq {z_\low}, \seq y \concatvar \seq {z_\high}) \wedge 
\formulaa_\low(\seq {y}\concatvar\seq{z_\low})\wedge
\formulaa_\high(\seq {y}\concatvar\seq{z_\high})) \vee \\ 
&\qquad\qquad\qquad  \exists {\seq y},{\seq {y'}}. (\formulaa_\low(\seq y \concatvar \seq {z_\low})\wedge
\formulaa_\high(\seq {y'}\seq {z_\high}) \wedge \orderformula_\low(\seq {y}, \seq {y'})) )\\
&\conserved(\seq {z_\low},\seq {z_\high}) =  \ \immutorder(\orderformula) \wedge (\exists \seq {y}. (\conservedsuper{\insup}(\seq y \concatvar \seq {z_\low}, \seq y \concatvar \seq {z_\high}) \wedge 
\formulaa_\low(\seq {y}\concatvar\seq{z_\low})\wedge
\formulaa_\high(\seq {y}\concatvar\seq{z_\high})) \vee \\ 
&\qquad\qquad\qquad  \exists {\seq y},{\seq {y'}}. (\formulaa_\low(\seq y \concatvar \seq {z_\low})\wedge
\formulaa_\high(\seq {y'}\seq {z_\high}) \wedge \orderformula_\low(\seq {y}, \seq {y'})) )
\end{align*}
\end{constructor}
Soundness of the domain linear sum constructor is established by combining the constructions used in \Cref{const:Pos} and \Cref{const:Lin}.

\begin{example}
    Continuing on our evaluation, we consider the Mutex Ring example.   
    In \Cref{sec:evaluation} we describe a ranking with structure DomLex(Lex(Bin,Bin,Bin). 
    Instead, one can consider a ranking with DomLin(Lex(Bin,Bin,Bin)) with $\formula(y) = \mathrm{token}(y)$. 
    In this way, if the token moves along the ring the rank is reduced, otherwise the rank is reduced when the local rank of the machine holding the token is reduced.
    While this is more natural, the same argument also holds for the lexicographic constructor (which captures many other cases as well). On the flip side, for the DomLin constructor we need to provide more hints as we have more existential quantifiers, and these hints are less natural. For this example, the  required hint is given by the term $\epsilon y. \mathrm{token}(y)$ (written with Hilbert's $\epsilon$ operator and axiomitized by $\exists y. \mathrm{token}(y) \to \mathrm{token}(\epsilon y. \mathrm{token}(y)) $).
\end{example}

%% file: appendices/proofs_constructors.tex
\section{Proofs of Soundness of Constructors}\label{appendix:proofs}
Here we give proofs of soundness of \Cref{const:Bin,const:Pos,const:DomLex,const:DomPW,const:DomPerm,const:Lex,const:PW}, \Cref{const:Lin,const:DomLin}.
Recall, our definition of soundness is that (i)~if the arguments of a constructor %
satisfy their assumptions, then the output of the constructor is an implicit ranking (for finite domains if the constructor is finite-domain), 
and (ii)~for constructors that receive implicit rankings as arguments, if at least one of the arguments is an implicit ranking for finite domains, so is the constructed implicit ranking.

\begin{claim}
    The binary constructor (\Cref{const:Bin}) is sound.    
\end{claim}
\begin{proof}
    \label{proof:binary_const}
    Let $\formula$ be a formula, and define $\rankname = \mathrm{Bin}(\formula)$. Fix a domain $\domain$ and define a ranking range for $\rankname,\domain $ by 
    $A=\{0,1\}$ with $0<1$. Clearly this defines a wfpo. We define a ranking function for $\domain$ by: $$f(s,\assign)=\begin{cases}
        1 & \struct,\assign\models \alpha(\seq x)\\ 
        0 & \text{otherwise} 
    \end{cases}$$
    Then, if $ \twopair\models \conserved(\seq {x_\low},\seq {x_\post})= 
    \formula_\low(\seq {x_\low}) \to \formula_\high(\seq {x_\high})$, then if 
    $\pairlow\models \alpha_\low(\seq {x_\low})$ we have $\pairhigh\models \alpha_\high(\seq {x_\high})$ as well and $f\pairlow=f\pairhigh=1$. 
    Otherwise, $f\pairlow = 0 \leq f\pairhigh$. 
    If $ \twopair\models     {\reduced}(\seq {x_\low},\seq {x_\high}) = \formula_\high(\seq {x_\high}) \wedge \neg \formula_\low(\seq {x_\low})$, then $f\pairlow = 0 < 1 = f(\pairhigh$

\end{proof}

\begin{claim}
    The position-in-order constructor (\Cref{const:Pos}) is sound.
\end{claim}

\begin{proof}
Let $\orderformula$ be a formula, let $\seq t$ be a sequence of terms and define $\rankname=\mathrm{Pos}(\seq t,\orderformula)$. Fix a domain $\domain$. Define $\mathrm{interp}(\orderformula,\domain) = \mathcal P(\assignset(\seq {y_0}\concatvar\seq {y_1},\domain) )$. Define a ranking range for $\rankname,\domain$ by $A = \mathrm{interp}(\orderformula,\domain)\times \assignset(\seq y,\domain)$ with the order $\leq$ on $A$ defined by: $(\interp_\low^\orderformula,\assign_\low) \leq (\interp_\high^\orderformula,\assign_\high)$ if $\interp_\low^\orderformula=\interp_\high^\orderformula$ and $\assign_\low \concatfunc \assign_\high \in \interp_\low^\orderformula$ or $\assign_\low = \assign_\high$.
For a structure $\struct$, the interpretation $\interp^\orderformula(\struct)$ of $\orderformula$ is the set of assignments $\assign_\low \concatfunc \assign_\high \in  \assignset(\seq {y_\low} \concatvar \seq {y_\high},\domain)$ such that $\struct,\assign_\low \concatfunc \assign_\high \models \orderformula$ when interpreting $\assign_\low$ as the assignment to $\Sigma_\low$ and $\assign_\high$ to $\Sigma_\high$ respectively.
A ranking function for $\domain$ is then defined by $f\pair=(\interp^\orderformula(\struct),\assign^{\seq t})$ where $\assign^{\seq t}$ is an assignment to $\seq y$ that takes its value based on the values $\assign$ gives to $\seq t$, that is $\assign^{\seq t}(y_i) = \bar \assign(t_i)$.

Then, if $ \twopair\models {\conserved}(\seq {y_\low} ,\seq {y_\high}) = \immutorder(\orderformula) \wedge (\orderformula_\low(\seq {t_\low},\seq {t_\high}) \vee \seq {t_\low} = \seq {t_\high})$.
Then, in particular $\interp^\orderformula(\struct_\low) = \interp^\orderformula(\struct_\high)$. We now wish to show that $ \assign_\low^{\seq t}\concatvar \assign_\high^{\seq t} \in \interp^\orderformula(\struct_\low)$ or $\assign_\low^{\seq t} = \assign_\high^{\seq t}$.
From the above, we have $\twopair\models  \orderformula_\low(\seq {t_\low},\seq {t_\high}) \vee \seq {t_\low} = \seq {t_\high}$ 
equivalently we have $(\struct_\low,\assign^{\seq t}_\low),(\struct_\high,\assign^{\seq {t}}_\high)\models  \orderformula_\low(\seq {y_\low},\seq {y_\high}) \vee \seq {y_\low} = \seq {y_\high}$ which correspond to $(\assign^{\seq t}_\low,\assign^{\seq t}_\high)\in \interp^\orderformula(\struct_\low)$ and $ \assign^{\seq t}_\low = \assign^{\seq t}_\high$ respectively.
For reduction: $\twopair\models {\reduced}(\seq {y_\low} ,\seq {y_\high})$ 
as before
$(\struct_\low,\assign^{\seq t}_\low),(\struct_\high,\assign^{\seq {t}}_\high)\models  \orderformula_\low(\seq {y_\low},\seq {y_\high})$, and so $(\assign^{\seq t}_\low,\assign^{\seq t}_\high)\in \interp^\orderformula(\struct_\low)$, from the order axioms we get $(\assign^{\seq t}_\high,\assign^{\seq t}_\low)\notin \interp^\orderformula(\struct_\low)$, and in particular $\assign^{\seq t}_\high \neq \assign^{\seq t}_\low$ so $f\pairlow<f\pairhigh$.
\end{proof}

\begin{claim}
    The pointwise constructor (\Cref{const:PW}) is sound. 
\end{claim}

\begin{proof}    
    Let $\rankname^i=(\conservedsuper{i},\reducedsuper{i})$ be implicit rankings for $i=1,\ldots,m$, each with parameters $\seq x$, and define $\rankname=\mathrm{PW}(\rankname^1,\ldots,\rankname^m)$.
    Let $\domain$ be a domain (if any of $\rankname^i$ are implicit rankings for finite domains, assume as well that $\domain$ is finite), 
    Let $(A_i,\leq_i)$ be a ranking range and $f_i$ be a ranking function of each $\rankname_i$. 
    A ranking range of $\rankname$ for $\domain$ can be defined by $(A_1\times \cdots \times A_m, \leq_{\text{pw}})$ and a ranking function can be defined by $f\pair = (f_1\pair,\ldots,f_m\pair)$.

    Indeed, if $\twopair\models {\conserved}(\seq {x_\low},\seq {x_\high}) = \bigwedge_i \conservedsuper{i}(\seq {x_\low},\seq {x_\high})$, then, for all $i\in [m]$ we have $\twopair\models  \conservedsuper{i}(\seq {x_\low},\seq {x_\high})$, by assumption we get  $f_i\pairlow \leq_i f_i\pairhigh $
    for every $i\in [m]$, and so $f\pairlow \leq_{\text{pw}} f\pairhigh$.
    For reduction, if $\twopair\models {\reduced}(\seq {x_\low},\seq {x_\high}) = {\conserved}(\seq {x_\low},\seq {x_\high}) \wedge \bigvee_i \reducedsuper{i}(\seq {x_\low},\seq {x_\high})$, 
    then $f\pairlow \leq_{\text{pw}} f\pairhigh$, 
    and there exists $i\in [m]$ such that $\twopair\models \reducedsuper{i}(\seq {x_\low},\seq {x_\high})$ then $f_i\pairlow < f_i\pairhigh$ and so $f\pairlow <_{\text{pw}} f\pairhigh$.
\end{proof}

\begin{claim}
    The lexicographic constructor (\Cref{const:Lex}) is sound. 
\end{claim}

\begin{proof}    
    Let $\rankname^i=(\conservedsuper{i},\reducedsuper{i})$ be implicit rankings for $i=1,\ldots,m$, each with parameters $\seq x$ and define $\rankname=\mathrm{Lex}(\rankname^1,\ldots,\rankname^m)$.
    Let $\domain$ be a domain (if any of $\rankname^i$ are implicit rankings for finite domains, assume as well that $\domain$ is finite), 
    Let $(A_i,\leq_i)$ be a ranking range and $f_i$ be a ranking function of each $\rankname_i$. 
    A ranking range of $\rankname$ for $\domain$ can be defined by $(A_1\times \cdots \times A_m, \leq_{\text{lex}})$ and a ranking function can be defined by $f\pair = (f_1\pair,\ldots,f_m\pair)$.
    
    Indeed, if $\twopair\models {\conserved}(\seq {x_\low},\seq {x_\high})  = \bigvee_i
(\reducedsuper{i}(\seq {x_\low},\seq {x_\high}) \wedge \bigwedge_{j<i} 
 \conservedsuper{j}(\seq {x_\low},\seq {x_\high})
)\vee \bigwedge_i \conservedsuper{i}(\seq {x_\low},\seq {x_\high})$. If the first disjunct is satisfied then there exists an $i\in [m]$ such that $f_i\pairlow<f_i\pairhigh$ and for all $j<i$ we have $f_j\pairlow\leq_j f_j\pairhigh$, giving $f\pairlow\leq_{\text{lex}} f\pairhigh$. In the other case, for all $i\in [m]$, $f_i\pairlow\leq_i f_i\pairhigh$ which also gives $f\pairlow\leq_{\text{lex}} f\pairhigh$. 
For reduction, we only need to consider the first case, and it follows that $f_i\pairlow < f_i\pairhigh$ so $f_i\pairlow \neq f_i\pairhigh$ as required.
\end{proof}

\begin{claim}
    The linear sum constructor (\Cref{const:Lin}) is sound. 
\end{claim}

\begin{proof}
Let $\rankname^i=(\conservedsuper{i},\reducedsuper{i})$ be implicit rankings for $i=1,\ldots,m$, each with parameters $\seq x$, let $\formula^1,\ldots,\formula^m$ be formulas with free variables $\seq x$, and define $\rankname=\mathrm{Lin}(\rankname^1,\ldots,\rankname^m,\formula^1,\ldots,\formula^m)$. 
Let $\domain$ be a domain (if any of $\rankname^i$ are implicit rankings for finite domains, assume as well that $\domain$ is finite), 
Let $(A_i,\leq_i)$ be a ranking range and $f_i$ be a ranking function of each $\rankname_i$. 
We define a ranking range for $\rankname,\domain$ by: $(\{\bot\} \uplus A_1\uplus \cdots\uplus A_m,\leq_\oplus)$
A ranking function for $\rankname$ first maps each  \pairname \ into a set $A_i$ (or $\bot$) according to the values of the formulas $\formulaa^1,\ldots,\formulaa^m$. The {\pairnames} mapped to $A_i$ are then evaluated by the corresponding ranking function $f_i$, resulting in:  
$$
f\pair = \begin{cases}
    f_i\pair &   \pair\models \formulaa^i\\
    \bot & \text{no such }i
\end{cases}
$$
This is well-defined as the $\formulaa^i$, defined by $\formulaa^i(\seq x) = \formula^i (\seq x) \wedge \bigwedge_{j<i} \neg \formula^j (\seq x)$,  are mutually exclusive.

Then, assume
$\twopair\models {\conserved}(\seq {x_\low},\seq {x_\high})=$
$ (\bigvee_{i} \conservedsuper{i}(\seq {x_\low},\seq {x_\high}) \wedge \formulaa_\low^i(\seq {x_\low}) \wedge \formulaa_\high^i(\seq {x_\high})) \vee (\bigvee_{i<j} \formulaa_\low^i(\seq {x_\low}) \wedge \formulaa^j_\high(\seq {x_\high}))$.
If the first disjunct is satisfied, then there exists $i$ such that $f_i\pairlow\leq_i f_i\pairhigh$, $\pairlow\models \formulaa_\low^i(\seq {x_\low})$ and $\pairhigh\models \formulaa_\high^i(\seq {x_\high})$. It follows that $f\pairlow\leq_{\oplus} f\pairhigh$.
If the second disjunct is satisfied, we have $i<j$ such that $\pairlow\models \formulaa_\low^i(\seq {x_\low})$ and $\pairhigh\models \formulaa_\high^j(\seq {x_\high})$, so $f\pairlow \in A_i$ and $f\pairhigh \in A_j$, it follows that 
$f\pairlow \leq_{\oplus} f\pairhigh$.
For reduction, in the first case we get instead  $f_i\pairlow <  f_i\pairhigh$ and in particular $f\pairlow\neq f\pairhigh$.
In the second case, $f\pairlow,f\pairhigh$ are in different disjuncts and so different.
\end{proof}

\begin{claim}
    The domain-pointwise constructor (\Cref{const:DomPW}) is sound. 
\end{claim}

\begin{proof}
    Let $\rankname^{\insup}$ be an implicit ranking with parameters $\seq x$, let $\seq y,\seq z$ such that $\seq x = \seq y \concatvar \seq z$, and define $\mathrm{DomPW}(\rankname^{\insup},\seq y)$.
    Consider a \textit{finite} domain $\domain$ and let $(A_\insup,\leq_\insup)$ and $f_\insup$ be a ranking range for $\rankname^\insup$. 
    We define a ranking range for $\rankname$ by $(\assignset(\seq y,\domain)\to A_\insup,\leq_{\text{pw}})$, and a ranking function by $f\pair=\lambda \otherassign\in \assignset(\seq y, \domain). f_\insup(\struct,\otherassign\concatfunc\assign)$, for any $\pair\in \pairset(\Sigma,\seq z,\domain)$.

    Then, assume
    $\twopair\models {\conserved}(\seq {x_\low},\seq {x_\high}) = \forall {\seq y}. \conservedsuper{\insup}(\seq y \concatvar \seq {z_\low}, \seq y \concatvar \seq {z_\high})$.
    Then for any assignment $\otherassign$ to $\seq y$ we have $(\struct_\low, \assign_\low),(\struct_\high,\assign_\high),\otherassign \models \conservedsuper{\insup}(\seq y \concatvar \seq {z_\low}, \seq y \concatvar \seq {z_\high})$, which by our conventions is equivalent to $(\struct_\low, \otherassign\concatfunc\assign_\low),(\struct_\high,\otherassign\concatfunc\assign_\high) \models \conservedsuper{\insup}(\seq {y_0} \concatvar \seq {z_\low}, \seq {y_1} \concatvar \seq {z_\high})$
    By the soundness of $\rankname^\insup$ we get 
    $f_\insup (\struct_\low,\otherassign \concatfunc \assign_\low) \leq_\insup f_\insup (\struct_\high,\otherassign \concatfunc \assign_\high)$. 
    This holds for any $\otherassign\in \assignset(\seq y,\domain)$ so it follows that: 
    $\lambda \otherassign. f_\insup(\struct,\otherassign\concatfunc\assign_\low) \leq_{\text{pw}} \lambda \otherassign. f_\insup(\struct,\otherassign\concatfunc\assign_\high)$ and so 
    $f\pairlow \leq_{\text{pw}} f\pairhigh$.
    For reduction, if we have additionally $\twopair \models (\exists \seq y. {\reducedsuper{\insup}}( \seq y \concatvar \seq {z_\low}, \seq y \concatvar \seq {z_\high} ))$ then by the same argument there is assignment $\otherassign\in\assignset(\seq y,\domain)$ such that
    $f_\insup (\struct_\low,\otherassign \concatfunc \assign_\low) <_\insup f_\insup (\struct_\high,\otherassign \concatfunc \assign_\high)$ it follows that $f\pairlow \ne f\pairhigh$.
\end{proof}

For a set $Y$, a partially-ordered set $A$, and two functions $a,b\in Y\to A$, we say that $a \leq_{\text{perm}} b$ if there exists a bijection $\bijec$ on $Y$ such that $a \leq_{\text{pw}} b\circ \bijec$. 
The result,  
$\leq_{\text{perm}}$, is a \emph{preorder} on $Y\to A$. It induces a partial order on the quotient set of $Y\to A$ with respect to the equivalence relation  $\equiv_{\text{perm}}$. This equivalence relation is derived from the preorder by $a\equiv_{\text{perm}}b \iff$ $a\leq_{\text{perm}} b$ and $b\leq_{\text{perm}} a$, equivalently, there exists $\bijec$ such that $a = b\circ \bijec$.
It is useful to observe that if $[a],[b]$ are two equivalence classes 
of $\equiv_{\text{perm}}$ such that $[a]<_{\text{perm}} [b]$ then there exists a bijection $\bijec$ such that $a <_{\text{pw}} b\circ \bijec$, equivalently, we can say there is $\bijec$ such that $a\circ \bijec <_{\text{pw}} b$

\begin{claim}
    If $(A,\leq_A)$ is a wfpo and $Y$ is finite, $\leq_{\text{perm}}$ is a wfpo on the quotient set $Y\to A$ modulo $\equiv_{\text{perm}}$.
\end{claim}
\begin{proof}
    Assume towards contradiction that we have a sequence of equivalence classes $[a_0],[a_1],\ldots$ such that  $[a_{i+1}] <_{\text{perm}} [a_i]$ for all $i\in \nat$. From the observation above, we have bijections $\bijec_1,\bijec_2,\ldots$ such that $a_{i+1}  \circ \bijec_{i+1} <_{\text{pw}} a_i$. It follows that for every $i\in \nat$ there is an element ${x_i}$ for which $a_{i+1}(\bijec_{i+1}(x_i)) < a_i({x_i})$.
    For every $y\in Y$ define a sequence of elements $y_0,y_1,\ldots$ by: $y_0 = y$ and $y_{i+1} = \bijec_i (y_i)$ for every $i\in \nat$. From the above we get that for every $y\in Y,i \in \nat$ $a_{i+1}(y_{i+1}) \leq a_{i}(y_i)$.
    Because $\bijec_i$ are all bijections the set of elements $\{y_i \mid y\in Y\}$ is $Y$ itself, for all $i \in \nat$.
    It then follows that for every $i\in \nat$ there exists $y\in Y$ such that $x_i = y_i$, for this $y$  we have $a_{i+1}(y_{i+1})<a_i(y_i)$.
    Let $\hat y$ be such that for infinitely many $i\in \nat$ we have $\hat{y}_i = x_i$.
    For $\hat y$, the sequence $a_0(\hat {y_0}),a_1(\hat {y_1}),\ldots$. From the above, in this sequence the order is conserved in every index and reduced infinitely many times - contradicting well-foundedness of $A$.
\end{proof}

Before proving soundness of the permuted-pointwise constructor we show our encoding of permutations is sound:

\begin{claim}
Given ${\seq y^1_{\scriptstyle\rightarrow}},\seq {y^1_{\scriptstyle\leftarrow}},\ldots,\seq {y^k_{\scriptstyle\rightarrow}},\seq {y^k_{\scriptstyle\leftarrow}}$ such that for every $1\leq i < j \leq k$ 
we have ${\seq y^i_{\scriptstyle\rightarrow}}\neq {\seq y^j_{\scriptstyle\rightarrow}}, {\seq y^i_{\scriptstyle\leftarrow}}\neq {\seq y^j_{\scriptstyle\leftarrow}},{\seq y^i_{\scriptstyle\rightarrow}}\neq {\seq y^j_{\scriptstyle\leftarrow}}
     $
there exists a permutation $\bijec$ such that for every $\seq y$, the term $\seq y_\bijec =  \mathrm{ite}(\seq y = \seq {y^1_{\scriptstyle\rightarrow}},\
\seq {y^1_{\scriptstyle\leftarrow}}, \
\mathrm{ite}(\seq y = \seq {y^1_{\scriptstyle\leftarrow}},\
\seq {y^1_{\scriptstyle\rightarrow}} , \ 
\ldots, \ 
\mathrm{ite}(\seq y = \seq {y^k_{\scriptstyle\rightarrow}}, \ \seq {y^k_{\scriptstyle\leftarrow}}, \ \mathrm{ite}(\seq y = \seq {y^k_{\scriptstyle\leftarrow}}, \ \seq {y^k_{\scriptstyle\rightarrow}} , \ \seq y))))$.
\end{claim}

\begin{proof}
    Define $\bijec$ by:
    $$\sigma(\seq y)=\begin{cases}
         \seq y^i_{\scriptstyle\leftarrow} & \exists i. \seq y = \seq y^i_{\scriptstyle\rightarrow} \\
        \seq y^i_{\scriptstyle\rightarrow} & \exists i. \seq y =  \seq y^i_{\scriptstyle\leftarrow} \\
        \seq y & \text{else}
    \end{cases}$$
    $\bijec$ is well-defined: for any value of $\seq y$, if there exists $i$ for which $\seq y =  \seq y^i_{\scriptstyle\rightarrow}$, then from assumption this $i$ is unique, so the first case may only return one value, if there also exists $j$ such that $\seq y =  \seq y^j_{\scriptstyle\leftarrow}$, then from assumption we get that $i=j$ and from both cases we get $\bijec(\seq y)=\seq y$, other cases for the proof follow similarly.
    It is easy to see similarly that $\bijec$ is a permutation.
\end{proof}

\begin{claim}
    The domain permuted-pointwise constructor (\Cref{const:DomPerm}) is sound. 
\end{claim}

\begin{proof}       
Let $\rankname^{\insup}$ be an implicit ranking with parameters $\seq x$, let $\seq y,\seq z$ such that $\seq x = \seq y \concatvar \seq z$, and define $\mathrm{DomPerm}(\rankname^{\insup},\seq y)$. Consider a \textit{finite} domain $\domain$ and let $(A_\insup,\leq_\insup)$ and $f_\insup$ be a ranking range for $\rankname^\insup$. 
We define a ranking range for $\rankname$ by $\assignset(\seq y,\domain)\to A_\insup$ modulo $\equiv_{\text{perm}}$ ordered by $\leq_{\text{perm}}$, and a ranking function by $f\pair=[ \lambda \otherassign\in \assignset(\seq y, \domain). f_\insup(\struct,\otherassign\concatfunc\assign) ]_{\equiv_{\text{perm}}}$, for any $\pair\in \pairset(\Sigma,\seq z,\domain)$.

Assuming, $\twopair \models \conserved(\seq {z_\low},\seq {z_\high}) = \tilde \exists\bijec. \
    \forall \seq {y}. \ \conservedsuper{\insup}(\seq y \concatvar \seq {z_\low}, \seq {y_\bijec} \concatvar \seq {z_\high} )$. 
    From the above, it follows that there is a permutation $\bijec$ on $\assignset(\seq y,\domain)$ such that for any $\otherassign\in \assignset(\seq y,\domain)$ $\seq y_\bijec$ is interpreted as $\sigma(\otherassign)(\seq y)$.
    Now, for any assignment $\otherassign\in \assignset(\seq y,\domain)$ we have  $(\struct_\low, \assign_\low),(\struct_\high,\assign_\high),\otherassign \models \conservedsuper{\insup}(\seq y \concatvar \seq {z_\low}, \seq {y_\bijec} \concatvar \seq {z_\high} )$, which by our conventions is equivalent to $(\struct_\low, \otherassign\concatfunc\assign_\low),(\struct_\high,\bijec(\otherassign)\concatfunc\assign_\high) \models \conservedsuper{\insup}(\seq {y_\low} \concatvar \seq {z_\low}, \seq {y_\high} \concatvar \seq {z_\high} )$
    which means $f_\insup (\struct_\low,\otherassign \concatfunc \assign_\low) \leq_\insup f_\insup (\struct_\high,\bijec(\otherassign) \concatfunc \assign_\high)$. This holds for all $\otherassign$ so we get $ \lambda \otherassign. f_\insup(\struct,\otherassign\concatfunc\assign) \leq_{\text{pw}} ( \lambda \otherassign. f_\insup(\struct,\otherassign\concatfunc\assign) )\circ \bijec$ and so $f\pairlow \leq_{\text{perm}} f\pairhigh$.
For reduction, we have very similarly, due to the equivalent characterisation of reduction above $f\pairlow <_{\text{perm}} f\pairhigh$.
\end{proof}

For a partially-ordered set $A$ with 
a partial order $\leq_A$ and a set $Y$ with a wfpo $\leq_Y$,
the set of functions $Y\to A$ can be ordered by the lexicographic partial order:
for any $a,b\in Y\to A$ we have 
$a\leq_{\text{lex}}~b \iff a(y) \leq_A b(y)$ for all minimal elements of the set $\{y\in Y \mid a(y)\neq b(y)\}$ (for $a\neq b$, this set is not empty, because $\leq_{Y}$ is a wfpo). Equivalently, for every $y\in Y$ such that $a(y)\not\leq_A b(y)$ there exists $y^*$ such that $y^* <_Y y$ with $a(y^*) <_A b(y^*)$.

\begin{claim}
    If $(A,\leq_A)$ is a wfpo and $Y$ is finite, $\leq_{\text{lex}}$ is a wfpo on $Y\to A$.
\end{claim}

\begin{proof}
    Assume towards contradiction that we have a sequence of functions in $Y\to A$, $a_0,a_1,\ldots$ such that $a_{i+1} <_{\text{lex}} a_i$ for all $i\in \nat$. 
    It follows that for every $i\in \nat$ we have $x_i\in Y$ such that $a_{i+1}(x_i) <_A a_i (x_i)$. 
    Denote by $X$ the set of elements $x\in Y$ such that $x=x_i$ for infinitely many $x$. 
    Let $\hat x$ be a minimal element of $X$ by the order $\leq_Y$ (guaranteed to exist from finiteness).
    We want to argue that from some point forward the sequence $a_i(\hat x)$ satisfies $a_{i+1}(\hat x) \leq_A a_i(\hat x)$ for all $i$.
    Indeed, there is some point from which forwards $a_{i+1}(x_i) <_A a_i(x_i)$ only for values in $X$ (again, using finiteness of $Y$). 
    Let $i$ large enough and assume towards contradiction that we have $a_{i+1}(\hat x) \not\leq a_i(\hat x)$, then from $a_{i+1} <_{\text{lex}} a_i$ we get an element $y \in Y$ such that $y <_Y \hat x$ and $a_{i+1}(y)<a_i(y)$ but this contradicts minimality of $\hat x$, as $y$ must also be in $X$ from the observation above.
    We have established that $a_{i+1}(\hat x)\leq_A a_i(\hat x)$ for all $i$ from some point forwards, but additionally, we have $a_{i+1}(\hat x)<_A a_i(\hat x)$ for infinitely many $i$ as $\hat x\in X$, so we get a contradiction to well-foundedness of $<_A$.
\end{proof}

\begin{claim}
    The domain-lexicographic constructor (\Cref{const:DomLex}) is sound. 
\end{claim}

\begin{proof}
Let $\rankname^{\insup}$ be an implicit ranking with parameters $\seq x$, let $\seq y,\seq z$ such that $\seq x = \seq y \concatvar \seq z$, and a formula $\orderformula$ over $\Sigma$ with free variables $\seq {y_\low},\seq {y_\high}$, and define $\mathrm{DomLex}(\rankname^{\insup},\seq y,\orderformula)$.
Consider a \textit{finite} domain $\domain$ and let $(A_\insup,\leq_\insup)$ and $f_\insup$ be a ranking range for $\rankname^\insup$.
Define $\mathrm{interp}(\orderformula)$ as above.     
We define a ranking range for $\rankname$ by $( \mathrm{interp}(\orderformula)\times (\assignset(\seq y,\domain)\to A_\insup),\leq)$ where $(\interp^\orderformula_\low,a_\low)\leq (\interp^\orderformula_\high,a_\high)$ if $\interp^\orderformula_\low = \interp^\orderformula_\high$ and $a_\low$ is smaller or equal to $a_\high$ by the lexicographic order on $\assignset(\seq y,\domain)\to A^{\insup}$ defined by $\interp^\orderformula_\low$.
This means for every $\otherassign \in \assignset(\seq y,\domain)$ such that $a_\low(\otherassign) \not\leq_\insup a_\high(\otherassign)$ there exists $\otherassign^*\in \assignset(\seq y,\domain)$ such that $\otherassign^*\concatfunc\otherassign\in \interp^\orderformula_\low$ with $a_\low(\otherassign^*) <_\insup a_\high(\otherassign^*)$.
It is easy to verify that this defines a wfpo.
We then define the ranking function by $f\pair = (\interp^\orderformula(\struct),\lambda \otherassign\in \assignset(\seq y, \domain). f_\insup(\struct,\otherassign\concatfunc\assign))$ 
where $\interp^\orderformula(\struct)$ is defined as above to be the set of assignments $\assign_\low \concatfunc \assign_\high \in  \assignset(\seq {y_\low} \concatvar \seq {y_\high},\domain)$ such that $\struct,\assign_\low \concatfunc \assign_\high \models \orderformula$.

Assume $\twopair \models {\conserved}(\seq {z_\low},\seq {z_\high}) = \immutorder(\orderformula)\wedge \forall \seq{y}.  (\conservedsuper{\insup}(\seq {y}\concatvar \seq {z_\low},
    \seq {y}\concatvar \seq {z_\high})
    \vee \exists \seq {y^*}. (
    \orderformula_\low(\seq{y^*},\seq{y})
    \wedge
    \reducedsuper{\insup}(\seq {y^*}\concatvar \seq {z_\low},
    \seq {y^*}\concatvar \seq {z_\high})
    ) 
    )$.
First, it follows that $\interp^\orderformula(\struct_\low)=\interp^\orderformula(\struct_\high)$. 
Then, for any $\otherassign\in \assignset(\seq y,\domain)$ we have either $\conservedsuper{\insup}(\seq {y}\concatvar \seq {z_\low},
    \seq {y}\concatvar \seq {z_\high})$ or $\exists \seq {y^*}. (
    \orderformula_\low(\seq{y^*},\seq{y})
    \wedge
    \reducedsuper{\insup}(\seq {y^*}\concatvar \seq {z_\low},
    \seq {y^*}\concatvar \seq {z_\high})
    ) $ satisfied.
    Now we show that $f\pairlow \leq f\pairhigh$: let $\otherassign$ be such that $f_\insup(\struct_\low,\otherassign\concatfunc\assign_\low) \not\leq_\insup f_\insup(\struct_\high,\otherassign\concatfunc\assign_\high)$, then the first disjunct is not satisfied, so there exists $\otherassign^*$ such that $
    \orderformula_\low(\seq{y^*},\seq{y})
    \wedge
    \reducedsuper{\insup}(\seq {y^*}\concatvar \seq {z_\low},
    \seq {y^*}\concatvar \seq {z_\high})
    $ is satisfied, so $\otherassign^* \concatvar\otherassign \in \interp^\ell(s_\low)$ and $f_\insup(\struct_\low,\otherassign^*\concatfunc\assign_\low) <_\insup f_\insup(\struct_\high,\otherassign^*\concatfunc\assign_\high)$
    as necessary.
    For reduction, we additionally require $\exists \seq y. {\reducedsuper{\insup}}( \seq y \concatvar \seq {z_\low}, \seq y \concatvar \seq {z_\high} )$ to be satisfied, which clearly guarantees inequality.

\end{proof}

\begin{claim}
    The domain-linear sum constructor (\Cref{const:DomLin}) is sound. 
\end{claim}

\begin{proof}       
Let $\rankname^{\insup}$ be an implicit ranking with parameters $\seq x$, let $\seq y,\seq z$ such that $\seq x = \seq y \concatvar \seq z$, and a formula $\orderformula$ over $\Sigma$ with free variables $\seq {y_\low},\seq {y_\high}$, and define $\mathrm{DomLin}(\rankname^{\insup},\seq y,\orderformula)$.
Consider a \textit{finite} domain $\domain$ and let $(A_\insup,\leq_\insup)$ and $f_\insup$ be a ranking range for $\rankname^\insup$.
Define $\mathrm{interp}(\orderformula)$ as above.     
We define a ranking range for $\rankname$ by $( \{\bot\} \uplus (\mathrm{interp}(\orderformula)\times \assignset(\seq y,\domain)\times A_\insup),\leq)$ where $(\interp^\orderformula_\low,\assign_\low,a_\low)\leq (\interp^\orderformula_\high,\assign_\high,a_\high)$ if $\interp^\orderformula_\low = \interp^\orderformula_\high$ and either $\assign_\low\concatfunc\assign_\high \in \interp^\orderformula_\low$ or $\assign_\low = \assign_\high$ and $a_0 \leq_\insup a_1$.
It is easy to verify that this defines a wfpo.
Define a ranking function by:
$$
f\pair = \begin{cases}    (\interp^\orderformula(\struct),\otherassign, f_\insup (\struct,\otherassign\concatfunc\assign)) &  (\struct,\otherassign\concatfunc\assign)\models \formulaa(\seq y\concatvar\seq z)\\
    \bot & \text{no such }\otherassign \text{ or not unique}
\end{cases}
$$
As long as $\pair$ satsifies the order axioms for $\orderformula$ (which is true for all interesting {\pairnames}), there is a unique $\otherassign$ for which $(\struct,\otherassign\concatfunc\assign)\models \formulaa(\seq y\concatvar\seq z)$. Recall that $\formulaa(\seq y \concatvar \seq z) = \formula(\seq y\concatvar\seq z) \wedge \forall {\seq {y''}}.( \orderformula(\seq {y},\seq {y''})\vee (\seq y = \seq {y''}) \vee \neg \formula(\seq {y''}\concatvar \seq {z}))$, so if it holds for two different assignments $\otherassign,\otherassign'$ then we will have $\otherassign\concatfunc\otherassign'$ and $\otherassign'\concatfunc\otherassign$ in the order interpretation, contradicting antisymmetry.

Now, assuming $\twopair\models \conserved(\seq {z_\low},\seq {z_\high})$, then as before, from $\immutorder(\orderformula)$ we have $\interp^\orderformula(\struct_\low) = \interp^\orderformula(\struct_\high)$. Again, we have two cases, in the first case: $\exists {\seq y},{\seq {y'}}. (\formulaa_\low(\seq y \concatvar \seq {z_\low})\wedge
\formulaa_\high(\seq {y'}\seq {z_\high}) \wedge \orderformula_\low(\seq {y}, \seq {y'})) )$ is satisfied. It follows that we have two assignments $\otherassign,\otherassign'$ such that $\otherassign\concatfunc\otherassign'\in \interp^\orderformula(\struct_\low)$ which guarantees $f\pairlow \leq f\pairhigh$ (and in fact, reduction).
In the other case, we have $\exists \seq {y}. (\conservedsuper{\insup}(\seq y \concatvar \seq {z_\low}, \seq y \concatvar \seq {z_\high}) \wedge 
\formulaa_\low(\seq {y}\concatvar\seq{z_\low})\wedge
\formulaa_\high(\seq {y}\concatvar\seq{z_\high}))$ satisfied. In this case, the second component of $f\pairlow$ and $f\pairhigh$ are equal to the same assignment $\otherassign$ and, for the third component we have $ f_\insup (\struct_\low,\otherassign\concatfunc\assign_\low)\leq_\insup f_\insup(\struct_\low,\otherassign\concatfunc\assign_\high)$ which is other case for $f\pairlow \leq f\pairhigh$. For reduction, the first case is the same, and in the second case we get $<_\insup$ instead of $\leq_\insup$ and so $f\pairlow \neq f\pairhigh$.
\end{proof}

%% file: appendices/extras_technical.tex
\section{Arithmetic Interpretation of Constructors}
\label{appendix:heights}
In the soundness proofs of our constructors we provide witnesses for the implicit rankings defined by the constructors by constructing a ranking range and a ranking function in a recursive manner using set-theoretic notions.
In this section we focus on the case of finite domains and show that in many cases we can construct more intuitive ranking functions that use simple arithmetic expressions, such as summations, and map states to the natural numbers. 

For programs that manipulate infinite domains, such as integers, it is commonly not possible to find a ranking function that maps to $\nat$ (see \cite{ordinal_valued_ranking}). 
However, for systems that are finite-state in each instance but unbounded, like all of our examples, this is possible. 
In fact, we show that for implicit rankings defined by our constructors we can systematically construct functions that map states to a finite interval $[0,\ldots,h]$ where $h$ is a function of $|\domain|=n$.
Such a ranking function also gives $h$ as a bound on the termination time --- the number of transitions (or, more generally, the number of visits to fairness) required until the eventuality is reached.

For simplicity of the presentation we present %
the construction of such ranking functions with two simplifying assumptions (both of which can be lifted):
\begin{itemize}
\item We assume that any system order $\orderformula$ used in constructors has a single argument and is a total order on the domain. Accordingly, we denote the finite domain $\domain = \{0,\ldots,n-1\}$, where the numbers reflect the order. %
\item We restrict domain-based constructors to aggregate over a single %
variable $y$.
\end{itemize}

For the binary constructor $\mathrm{Bin}(\formula)$, the ranking function we define is an indicator $f\pair = 1[\pair \models\formula]$, with bound $h=1$.
For a Position-in-Order constructor $\mathrm{Pos}(t,\orderformula)$
the construction is much simplified to $f\pair = \assign(t)$ with bound $h=n-1$.

For aggregating rankings $\rankname^1,\ldots,\rankname^m$, assume they correspond to ranking functions $f_1,\ldots,f_m$ with bounds $h_1,\ldots,h_m$.
$\mathrm{PW}(\rankname^1,\ldots,\rankname^m)$ can be interpreted by the ranking function 
$f\pair = \sum_{i=1}^m f_i\pair$ with bound 
$\sum_{i=1}^m h_i$, and $\mathrm{Lex}(\rankname^1,\ldots,\rankname^m)$ can be interpreted by the ranking function 
$f\pair = \sum_{i=1}^m  (\prod_{j=i+1}^m h_j ) f_i\pair$
with bound $\prod_j (h_j+1) - 1 $.
For example, if $m=2$, we have $\mathrm{PW}(\rankname^1,\rankname^2)$ interpreted by $f\pair = f_1\pair + f_2\pair$ with bound  $h_1+h_2$, and $\mathrm{Lex}(\rankname^1,\rankname^2)$ interpreted by $f\pair = (h_2+1) f_1\pair + f_2\pair$ with bound $(h_1+1)(h_2+1)-1$.

For domain-based aggregations, assume the given implicit ranking $\rankname$ corresponds to a ranking function $f(\struct,v_y\concatfunc v_{\seq z})$ with bound $h$.
Both constructors $\mathrm{DomPW}(\rankname,y)$ and $\mathrm{DomPerm}(\rankname,y)$ can be interpreted by the ranking function $f(\struct,v_{\seq z}) = \sum_{i=0}^{n-1}
f(\struct,[y \mapsto i] \concatfunc v_{\seq z})$ (where $[y \mapsto i]$ denotes an assignment of $i$ to $y$) with bound $n\cdot h$. 
This is because both pointwise reduction and permuted pointwise reduction guarantee a reduction in the sum. (In fact, we could choose any other monotonic, commutative, associative operation, such as product, but summation gives a tight bound.)
Finally, $\mathrm{DomLex}(\rankname,y)$ can be interpreted by $\sum_{i=0}^{n-1} (h+1)^{n-1-i} f(\struct,[y \mapsto i] \concatfunc v_{\seq z} )$, which gives a tight bound of $(h+1)^n-1$.

\begin{example}
We now demonstrate this idea for both motivating examples of \Cref{sec:motivating_examples}:\begin{itemize}
    \item For \Cref{example:toystab} we use the implicit ranking defined by  $\mathrm{DomPerm(DomPW(Bin}(\formula))$ for $\formula(i,j) = \mathrm{priv}(i)\wedge \mathrm{lt}(i,j)$, this implicit ranking can be interpreted by $\sum_{i} \sum_{j} 1[\mathrm{priv}(i)\wedge \mathrm{lt}(i,j)]$ giving a bound of $n^2$ which is a tight bound on the complexity (up to a constant).
    \item For \Cref{example:binary_counter} we use the implicit ranking defined by Lex(DomLex(Bin),Pos), it can be interpreted by the expression $n\cdot\sum_{i=0}^{n-1}(2^{n-1-i} 1[a(i)]) + ptr$ bounded by $n\cdot 2^n -1 $ which is a tight bound on the complexity (up to a constant).
    \end{itemize}
\end{example}

%% file: appendices/proof_rule.tex
\section{General Proof Rule }\label{appendix:proof_rule}

In this section we present the proof rule we use for our implementation, it can be seen as a generalization of the proof rule present in \Cref{subsec:using_ranking}, or as a non-temporal simplification of the proof rule from \cite{towards_liveness_proofs}.
This a proof rule for general response properties under parameterized fairness assumptions. These are properties of the form $\prop=(\bigwedge_i \forall \seq x \globally \eventually r_i(\vec x)) \to  \globally (p \to\eventually q)$. 
A transition system $\Tspec$ satisfies $\prop$ if for every trace of $\Tspec$ that satisfies $r_i(\seq x)$ for each assignment to $\seq x$ infinitely often and every state in such trace that satisfies $p$ there follows a state in that trace that satisfies $q$.

We can prove a transition system $\Tspec=(\Sigma,\iota,\tau)$ satisfies such a property by finding a \emph{closed} implicit ranking $\rankname=(\conserved,\reduced)$, closed formulas $\globalinv$ and $\triggerinv$ and a formula $\helpful$ and validating:
\begin{enumerate}
    \item $\iota \to \globalinv$ \label{item:init_gen}
    \item $\globalinv \wedge \tau \to \globalinv'$ \label{item:consec_gen}
    \item $\globalinv \wedge p \wedge \neg q \to \triggerinv$ \label{item:trigger} 
    \item $\triggerinv \wedge \tau \wedge \neg q' \to \triggerinv'$ \label{item:stability}
    \item $\triggerinv \wedge \tau \wedge \neg q' \to \tilde\varphi_\leq$ \label{item:conserved_gen}
    \item $\triggerinv \to \bigvee_{i=1}^n \exists \vec x. \helpful_i(\vec x)$ \label{item:helpful}
    \item $\triggerinv \wedge \tau \wedge \neg q' \wedge  \helpful_i(\vec x) \wedge \neg r_i(\vec x) \to \helpful_i'(\vec x)$ for all $i$ \label{item:psi_stability}
    \item     
     $\triggerinv \wedge \tau \wedge \neg q' \wedge \helpful_i(\vec x) \wedge r_i(\vec x) \to \tilde\varphi_<  $ for all $i$\label{item:reduced_gen}
\end{enumerate}
where $\tilde \varphi_\leq$ and $\tilde \varphi_<$ are the same as $\conserved$ and $\reduced$ with $\Sigma_\low$ substituted by $\Sigma'$ and $\Sigma_\high$ substituted by $\Sigma$.
The motivation for these premises is similar to that of the proof rule presented in \Cref{subsec:using_ranking}, with minor changes: $\globalinv$ is an invariant that characterises all reachable states, $\triggerinv$ is a formula that characterises all reachable states on a trace after visiting $p$ and not yet visiting $q$.
Finally, $\helpful_i(\seq x)$ is a formula (called in \cite{towards_liveness_proofs} a stable scheduler) that characterises when a visit to fairness constraint $r_i(\seq x)$ should reduce the ranking. We then verify in \cref{item:helpful} that some $\helpful_i$ holds for some assignment to $\seq x$ and in \cref{item:psi_stability} that $\helpful_i(\seq x)$ is stable if we have not yet seen $r_i(\seq x)$.
We can get the proof rule presented in \Cref{subsec:using_ranking}, as a special case by considering a single, non-parameterized fairness assumption $r$ and setting $p = \iota, \rho = \mathrm{true}, \psi_i(\seq x) =\mathrm{true}$. We show the soundness of the proof rule below.

\begin{proof}
    Let $\pi$ be a trace of $\Tspec$, and let $\domain$ be the domain of $\pi$. We have $\pi(0)\models \iota$ and for every $\ell\in \nat$ we have $\pi(\ell),\pi(\ell+1)\models \tau$. Assume $\pi \models \bigwedge_i \forall \vec x \globally \eventually r_i(\vec x)$, and $k$ is some point where $\pi(k)\models p$, and assume towards contradiction that for all $\ell \geq k$ we have $\pi(\ell)\nvDash q$.
    Let $(A,\leq)$ be a ranking range for $\domain$ and let $f$ be a ranking function for $\domain$. As $\rankname$ is closed we can write $f:\structset(\Sigma,\domain)\to A$ and we have $(\struct,\struct')\models\tilde\varphi_\leq\implies f(\struct')\leq f(\struct)$ and $(\struct,\struct')\models \tilde\varphi_<(\struct,\struct')\implies f(\struct')<f(\struct)$. We will achieve contradiction by showing that the sequence of indices $k_0,k_1,k_2,k_3,\ldots$ such that $f(\pi(k_0)),f(\pi(k_1)),f(\pi(k_2)),\ldots$ is an infinitely decreasing sequence, contradicting $\leq$ being a wfpo.

    First, from \cref{item:init_gen,item:consec_gen} we have $\pi(k) \models \rho$, then from from \cref{item:trigger,item:stability} and the assumption we do not reach $q$ we get that for all $\ell \geq k$ we have $\pi(\ell)\models \triggerinv$.
    It then follows from \cref{item:conserved} that for every $\ell \geq k$ that $\pi(\ell),\pi(\ell+1)\models \conserved$, thus $f(\pi(\ell+1))\leq f(\pi(\ell))$. It follows that    $f(\pi(k)),f(\pi(k+1)),f(\pi(k+2)),\ldots$ is a weakly decreasing sequence, then it suffices to show a sequence of points $k_0,k_1,\ldots$ such that $f(\pi(k_{j}+1)<f(\pi(k_j))$ for all $j$. 
    
    Now we construct the sequence $k_0,k_1,\ldots$. Starting at $k$, from \cref{item:helpful} we have $\pi(k) \models \vee_i \exists\seq  x.\helpful_i(\seq x)$, 
    so there exists some assignment $\assign \in \assignset(\seq x,\domain)$ such that $\pi(k),\assign \models \helpful_i(x)$.
    From the fairness assumption there exists $k_0 \geq k$ such that $\pi(k_0)\models r_i(\seq x)$, take $k_0$ to be the minimal such number.
    for all $k \leq \ell \leq k_0$ we have, by induction on $\ell$, from \cref{item:psi_stability} $\pi(\ell),\assign \models \helpful_i(\seq x)$. Then, $\pi(k_0),\assign \models \helpful_i(x)\wedge r_i(x)$ and so from \cref{item:reduced_gen} $\pi(k_0),\pi(k_0+1)\models\reduced$ so $f(\pi(k_0+1))<f(\pi(k_0))$.
    To get $k_{j+1}$ repeat the argument, starting at $k=k_j$.    
\end{proof}

%% file: appendices/deadlock_freedom.tex
\section{Totality}\label{appendix:totality}

One common property that is often verified in liveness proofs is totality (in \cite{lvr} this is called deadlock freedom, although this term has other meanings). Totality in its most simple form says that for any reachable state $\struct$ there exists a state $\struct'$ such that there is a transition from $\struct$ to $\struct'$.
When totality does not hold it is possible that a liveness property is not satisfied because there is some reachable state from which there is no transition. 

In the proof rule we present in \Cref{appendix:proof_rule} there is no premise that validates that totality of 
$\Tspec=(\Sigma,\iota,\tau)$. This is because we consider liveness properties that are proven under fairness assumptions ($\globally\eventually r$). 
In this case, we show the liveness property only for traces that satisfy the fairness assumption, and in particular are infinite. If the system is not total for some state $\struct$ this state cannot be reachable on a fair trace and is thus inessential for proving the property.
In this sense, totality should be considered not for the liveness proof but for justifying the fairness assumptions.

For the examples we considered that do not have in-built fairness assumptions we added the fairness assumption $\globally\eventually \mathrm{true}$, which is equivalent to totality. 
In this sense, we assume totality and do not prove it. 
We want to emphasis that for all our examples totality does in fact hold and our assumption is sound.
The reason we do not verify totality is that totality is not directly first-order verifiable, as a direct encoding would require second-order quantification.

One way to verify totality is to first impose a syntactic restriction on the transition formula: $\tau = \theta \wedge \eta$ where $\theta$ is a formula over $\Sigma$ and $\eta$ is a formula over $\Sigma\uplus \Sigma'$ that only allows assignments to $\Sigma'$ symbols by terms over $\Sigma$ (This can be generalized to other syntactic restrictions).
Then for any state that satisfies $\theta$ there is a transition to some state (defined by the assignments), which then reduces verifying totality to verifying that $\theta$ is a safety property of $\Tspec$. 
A similar idea is presented and used in \cite{lotan2024provingcutoffboundssafety}.
In \cite{lvr} they essentially prove totality in this manner, but they do not restrict $\eta$ to be only assignments so totality is not compeletely guaranteed. 
On the other hand, they consider, like us, liveness properties under fairness, so proving totality is also not necessary.

%% file: appendices/abstractions.tex
\section{Abstractions}\label{appendix:abstractions}

In the proof for Dijkstra's $4$-state and $3$-state protocols, Ghosh's $4$-state protocol we used a form of abstraction to simplify the quantifier structure of our satisfiability calls. 
The natural first-order modeling for this protocol is with $\mathrm{prev}$ and $\mathrm{next}$ function symbols. 
With these symbols in the system our verification conditions fall outside of the decidable fragment EPR. 
While this does not always cause the queries to timeout, it caused a timeout for our more complicated queries, even when using hints for existential quantifiers in the implicit ranking. 

To account for this we used an abstraction method, similar to what is suggested in \cite{paxos_made_epr,quantifier_instantitation} which replaces derived definitions that use $\mathrm{next},\mathrm{prev}$ by new symbols.
This can be thought of as instrumentation variables that are axiomitized to hold the correct derived values for relevant machines. 
Additionally, we add as axioms properties of the concrete system that require the next and prev system for their proof.
While the instrumented variables are not guaranteed to hold the correct values for all machines, we can show that the abstracted transition system simulates the original system, which justifies the abstraction. 
In fact, the simulation can be shown using first-order reasoning (see \cite{lotan2024provingcutoffboundssafety}).

Additionally, some of the other premises of our proof rule (\Cref{appendix:proof_rule}) that do not involve the implicit ranking are valid only finite domains.
To use the finiteness assumption for these we employ the use of induction axioms~\cite{infinite_models,induction_axioms} which eliminate spurious infinite counterexamples.